%% file: zecale.tex
\documentclass[runningheads]{llncs}

\input{./config/packages.tex}
\input{./config/macros.tex}

\graphicspath{{./images/}}

\definecolor{lightgray}{rgb}{0.9,0.9,0.9}
\title{Zecale: Reconciling Privacy and Scalability on Ethereum}
\author{Antoine Rondelet}
\institute{Clearmatics, UK \newline
  \email{ar@clearmatics.com}}

\begin{document}
\maketitle

\begin{abstract}
    In this paper, we present \zecale{}, a general purpose SNARK proof aggregator that uses recursive composition of SNARKs. We start by introducing the notion of recursive composition of SNARKs, before introducing \zecale{} as a privacy preserving scalability solution. Then, we list application types that can emerge and be built with \zecale{}. Finally, we argue that such scalability solutions for privacy preserving state transitions are paramount to emulate ``cash'' on blockchain systems.
\end{abstract}

\section{Introduction}

As blockchain systems gained in popularity, several challenges have arisen revealing some of the current limitations of the technology. While privacy is a known issue on blockchains such as Bitcoin~\cite{bitcoin-whitepaper} and Ethereum~\cite{ethereum-whitepaper}, scalability is another important concern.
In fact, by their very nature of ``append only ledgers'', the more users transact on a blockchain, the more data is added to the state over which validators reach consensus. This increase in blockchain data, if not controlled, can lead to fewer validator nodes (and full nodes) yielding more and more centralization in the system. In addition to that, sudden increase in transaction volumes can dramatically inflate the transaction cost, as witnessed in late 2017 after the raise in popularity of the CryptoKitties game\footnote{\url{https://www.cryptokitties.co/}} that congested the Ethereum network\footnote{\url{https://www.coindesk.com/cat-fight-ethereum-users-clash-cryptokitties-congestion}} and triggered a sharp increase in gas price.
Even today, the important amount of transactions on public blockchains pushes the fees up, making networks like Ethereum less attractive to certain users\footnote{See these tweets for instance: \url{https://twitter.com/intocryptoast/status/1263372756625702913?s=21}, \url{https://twitter.com/sassal0x/status/1264555874992848898?s=21} or \url{https://twitter.com/fennie_wang/status/1266083935093583877?s=21} for instance.}.
Over the past years, several solutions have been developed in order to improve blockchain scalability. Some approaches rely on a technique called ``sharding'' which consists in splitting the entire state of the network into partitions denoted ``shards''. Doing so allows to move away from the ``all miners verify all transactions'' approach by having several network partitions validating transactions concurrently~\cite{sharding-blog-post,eth-sharding}.
Other projects such as \coda~\cite{cryptoeprint:2020:352} introduce the notion of  ``succinct blockchain'', which is built through the use of recursive composition of SNARKs (see~\cref{subsection-recursion}). In \coda~the entire transaction history is replaced by an argument\footnote{i.e.~a computationally sound proof} of computational integrity certifying that the state of the blockchain is valid. As such, any node willing to join the network only needs to verify one (small) argument (a SNARK) instead of going through and verifying the entire transaction history.

Similar techniques have been investigated to develop ``layer 2'' scaling solutions on Ethereum, and constructions such as \emph{``ZK-Rollups''}~\cite{zk-rollups,barry-rollup,zk-sync} have gained tremendous interest\footnote{While the name ``zk-rollup'' is wide-spread in the community, we find this name a bit confusing. In fact, \emph{zero-knowledge} is not strictly required (and generally not used at all in ``zk-rollup'' projects). A better definition of ``rollups'' may be ``proofs of computational integrity of the verification of a set of transactions'', but this long description diverges from the agreed upon and wide-spread terminology. As such, we will stick to the agreed upon vocabulary and use the ``zk-rollup'' term in this paper. We gently warn the reader than this term can be misleading.}.

Finally, as witnessed in the Bitcoin community, \emph{relay networks} have been developed as a way to improve blockchains scalability through better block propagation~\cite{10.1145/3340422.3343640,falcon-btc} (see also~\cite{fibre-btc,bip-152}). In fact, in~\cite{10.1145/3340422.3343640} the authors explain that the use of a relay network along blockchain systems like Bitcoin can help improving the throughput of the system by shortening the block interval --- by speeding up the block propagation while avoiding an increase in the block orphan rate. We note however, that, while improving the transaction processing time, such solutions do not allow to reduce the size of the blockchain.

\medskip

\paragraph{Our contribution.}

In this paper we present \zecale{} --- a general purpose aggregator using recursive composition of SNARKs that allows to improve the scalability of SNARK-based applications on Ethereum via aggregation of transactions off-chain. We show how the privacy solution \zeth{}~\cite{DBLP:journals/corr/abs-1904-00905} is complemented by \zecale{}, and how both solutions can be used to implement digital cash systems.

To do so, we will consider the use of a relay network as a way to gain sender anonymity when coupled with privacy-preserving protocols such as \zeth{}. Additionally, we will show how scaling solutions like \zecale{} can be deployed on nodes constituting the relay network in order to off-load state transaction verification work --- normally carried out by miners --- to relays, and thus enabling to keep the transaction history and the blockchain data condensed.

We note that the focus of this study is articulated around the Ethereum~\cite{ethereum-whitepaper} blockchain. Nevertheless, all results presented here are directly applicable to blockchains supporting smart-contracts deployment and equipped with a Turing-complete execution environment.

\section{Background and Notations}

Let \ppt~denote probabilistic polynomial time.
Let $\secpar \in \NN$ be the security parameter. We assume that all algorithms described receive as an implicit parameter the security parameter written in unary, $\secparam$.
All algorithms are modeled as Turing machines. An efficient algorithm is a probabilistic Turing machine running in polynomial time. All adversaries are modeled as efficient algorithms.
An adversary and all parties they control are denoted by letter $\adv$.
We write $\negl[]$ (resp.~$\poly[]$) to denote a negligible (resp.~polynomial) function.
Following~\cite{DBLP:journals/iacr/Shoup04,DBLP:conf/eurocrypt/BellareR06}, we structure security proofs as sequences of games. As such, we say that the adversary wins game $\game{GAME}$ defined on $X$ if they make $\game{GAME}$ return $1$.
We say that $\adv$ wins $\game{GAME}$ with negligible advantage if $\advantage{\game{GAME}}{\adv, X} \leq \negl$.

By the expression $\cardinality{x}$ we denote the length of $x$ if $x$ is a string, the cardinality of $x$ if $x$ is a set or the number of coordinates in $x$ if $x$ is a vector.
We denote by $[n], n \in \NN$ the set $\{0, \ldots, n-1\}$. We represent string concatenation with the $|| \colon \Lambda^l \times \Lambda^m \to \Lambda^{l+m}$ infix operator, where $\Lambda$ is an alphabet.
Additionally, we denote by $x \gets y$ the operation that assigns to $x$ the value of $y$. Likewise, we represent by $x \sample X$ the operation that assigns to $x$ a value selected uniformly at random from the finite set $X$.

Finally, we denote by $[] \colon S^n \times [n] \to S$ the infix operator that takes a tuple and an index as inputs and returns the element at the given index in the tuple. We use the notation $x[i]$ as syntactic sugar for $x[]i$ (e.g. $x \gets (A,7,U),\ x[2] = U$).

\subsubsection{Bilinear groups.}
\newcommand{\RomanNumeralCaps}[1]{\MakeUppercase{\romannumeral #1}}

Let $\paramGen$ be a bilinear group generator taking $\secparam$ as input and returning a tuple $(\GRPord, \GRP_1, \GRP_2, \GRP_T, \pair)$, where $\GRP_1, \GRP_2, \GRP_T$ are cyclic groups of prime order $\GRPord$, and $\pair$ is a map $\pair \colon \GRP_1 \times \GRP_2 \to \GRP_T$, such that $\pair$ is non-degenerate and bilinear.
We further assume that arithmetic in the groups and computing $\pair$ is efficient, that $\GRP_1 \neq \GRP_2$, and that there does not exist \emph{efficiently computable} homomorphisms between the two source groups (pairing of type \RomanNumeralCaps{3}~\cite{DBLP:journals/dam/GalbraithPS08}).
We denote by $\ggen_1, \ggen_2, \ggen_T$ the generators of the groups $\GRP_1, \GRP_2, \GRP_T$ respectively. We denote by $+$ the group operation in the two source groups $\GRP_1,\GRP_2$. Additionally, we define the infix operator $\cdot \colon \FF_r \times \GRP_{1,2} \to \GRP_{1,2}$ that represents the successive application of the group operation - e.g.~$k \cdot \ggen_1 = \ggen_1 + \cdots + \ggen_1$ ($k$ times). The group operation for the target group $\GRP_T$ is denoted by $*$, and the ``exponentiation'' operator $\hat{} \colon \GRP_T \times \FF_r$ denotes its successive application - e.g.~${\ggen_T}^k = \ggen_T * \ldots * \ggen_T$ ($k$ times).

\section{Preliminaries}

\subsection{Brief overview of $\IP$, $\nizk$ and \pcalgostyle{SNARK}}

In contrast with standard mathematical proofs, Goldwasser, Micali and Rackoff introduced the notion of \emph{zero-knowledge proofs}. These proofs enable a prover \prover~to prove a theorem to a verifier \verifier~without revealing anything other than the fact that the theorem is correct~\cite{DBLP:conf/stoc/GoldwasserMR85}.

In their seminal work, the authors focused on interactive protocols, where both \prover~and \verifier~are modeled as interactive Turing machines communicating by sharing their reading and writing tapes.
In brief, an interactive proof is a procedure by which a prover wants to prove a theorem to a verifier. To that end, the verifier is allowed to flip coins, and use these coin flips to ask questions (i.e.~``send challenges'') to the prover. The prover answers the verifier's questions, and after several interactions, the verifier either accepts or rejects the theorem.

Informally, the complexity class admitting interactive proofs ($\IP$) is defined as the class of languages $\LAN$ with the following properties:
\begin{description}
    \item [Completeness:]
        \[
            \exists \prover\ \suchthat \condprob{\verifier.\verify(x, \pi) = \true}{
                \begin{aligned}
                x \in \LAN \\
                \pi \gets \transcript(\prover, \verifier)
                \end{aligned}
            } \geq \frac{2}{3}
        \]
    \item [Soundness:]
        \[
            \forall \prover,\ \condprob{\verifier.\verify(x, \pi) = \true}{
                \begin{aligned}
                x \not\in \LAN \\
                \pi \gets \transcript(\prover, \verifier)
                \end{aligned}
            } \leq \frac{1}{3}
        \]
\end{description}

where $\transcript(\prover, \verifier)$ is the set of messages exchanged between the prover and the verifier, and where $\verify$ is the verification algorithm ran by the verifier to either accept or reject $x$.
Note that interaction is a very efficient tool for \emph{soundness amplification} as multiple executions of an interactive protocol can be used to bring the soundness error down\footnote{$k$ repetitions allow to bound the soundness error by $\frac{1}{3^k}$}.

Informally, we say that the protocol is \emph{zero-knowledge} if the verifier only learns the validity of the theorem being proven and nothing else.

Importantly, the notion of $\IP$ as per~\cite{DBLP:conf/stoc/GoldwasserMR85} does not make any computational assumption on the prover (modeled as ``all powerful''), while the verifier is assumed to have limited resources (i.e.~be $\ppt$). Likewise, the computational model assumes that both \prover~and \verifier~have a random tape, but that the verifier's coin flips are \emph{private} (i.e.~private coin).

This restriction on coin tosses contrasts with the notion of Arthur-Merlin (\AM) protocols introduced by Babai~\cite{DBLP:conf/stoc/Babai85,Babai1988ArthurMerlinGA}, in which the verifier's coin tosses are public, and thus accessible to the prover. In fact, in this model all the verifier's messages can be replaced by the output of a random beacon~\cite{MR731318} for instance.
Shortly after, Goldwasser and Sipser~\cite{DBLP:journals/acr/GoldwasserS89} showed however, that ``private coin tossing'' is no stronger than ``public coin tossing'' and their work led to the well acclaimed $\IP = \AM$ result.

Interestingly, Shamir showed that $\IP = \PSPACE$~\cite{DBLP:journals/jacm/Shamir92}, demonstrating the great power of randomness and communication\footnote{a simplified proof was published by Shen~\cite{DBLP:journals/jacm/Shen92}}. While $\IP$ is demonstrated to be a wide complexity class\footnote{we know that $\npol \subseteq \PSPACE$}, Goldreich and Oren showed~\cite{DBLP:conf/focs/Oren87,DBLP:journals/joc/GoldreichO94} that removing interactions between prover and verifier while preserving zero-knowledge yields a collapse to $\bpp$ in the plain model.\footnote{note that the word ``collapse'' here reflects our current understanding of the complexity class hierarchy. In fact, the relation between $\pol$ and $\bpp$ and the relation between $\pol$ and $\npol$ is not yet fully understood. As of today, it is believed that $\pol = \bpp$~\cite{DBLP:books/sp/goldreich2011/Goldreich11g,DBLP:conf/stoc/ImpagliazzoW97,DBLP:journals/eatcs/ClementiRT98} and that $\pol \neq \npol$.}

As a way to compensate the removal of interactions, Blum, Feldman and Micali~\cite{DBLP:conf/stoc/BlumFM88}, and then Blum, De Santis, Micali and Persiano~\cite{DBLP:journals/siamcomp/BlumSMP91} studied the notion of \emph{non-interactive zero-knowledge proofs}, in which prover/verifier communications and randomness are substituted by a shared \emph{common random string}.
This setting, generalized to the \emph{common reference string} (CRS) model, and declined in various flavors where, for example, the reference string is assumed to have a specific structure\footnote{in which case we talk about ``Structured Reference String'' (SRS)}, has been used to design a wide range of non-interactive zero-knowledge proof systems ($\nizk$).

We further note that, moving to the non-interactive setting allows to obtain \emph{publicly verifiable} proofs, which is of great interest for various cryptographic applications (e.g.~blockchain protocols).

\subsubsection{\pcalgostyle{SNARK}.}

We now focus on zk-SNARKs, a special type of $\nizk$.
Informally, a (zero-knowledge) Succinct Non-interactive Argument of Knowledge (SNARK), is a small proof of knowledge, which is non-interactive and sound against computationally bounded adversaries. Not surprisingly, the appealing performances of $\nizk$s like SNARKs take their source in the shift from an ``all powerful prover'' (like in $\IP$) to a ``computationally restricted'' prover; allowing protocol designers to rely on cryptographic assumptions to design efficient protocols.

The first known SNARK construction, proposed by Micali~\cite{DBLP:conf/focs/Micali94} coined \emph{Computationally Sound (CS) Proofs} is based on Kilian's construction~\cite{DBLP:conf/stoc/Kilian92} and proven secure in the Random Oracle model.
Following, various SNARK constructions have been established over the past decades (e.g~\cite{DBLP:conf/eurocrypt/GennaroGP013,DBLP:conf/asiacrypt/Groth10,DBLP:conf/tcc/Lipmaa12}). While different SNARK constructions offer different trade-offs, we will be focusing here on the pairing-based SNARK construction~\cite{DBLP:conf/eurocrypt/Groth16} ``compiled from'' \emph{Linear Interactive Proofs} (LIP) as proposed by Bitansky et al.~\cite{DBLP:conf/tcc/BitanskyCIPO13}.

\begin{remark}
    Our focus on~\cite{DBLP:conf/eurocrypt/Groth16} is justified by the fact that this scheme admits very small arguments ($3$ group elements), which, when used along a blockchain system, allows to minimize the size of the transactions. We note however that our result can be used with any other ``NP-complete'' SNARK schemes.
\end{remark}

Below, we provide an informal definition of a SNARK, and invite the reader to consult e.g.~\cite{DBLP:conf/crypto/GrothM17} for precise definitions.

\paragraph{Informal definition of SNARKs.}

Let $\REL \subset \bin^* \times \bin^*$ be a polynomial-time decidable binary relation.
We assume that $\secpar$ is explicitly deductible from the description of $\REL$.
Let $\LAN = \{\inp\ |\ \exists \wit\ \suchthat\ \REL(\inp, \wit)\}$ be the $\npol$ language defined by $\REL$. Here, $\inp$ is the instance\footnote{also referred to as ``public/primary input''} and $\wit$ is the witness.
Roughly speaking, $\snark$ is a publicly verifiable zero-knowledge Succinct Non-interactive Argument of Knowledge (zk-SNARK) if $\snark$ comports four $\ppt$ algorithms $\kgen, \prover, \verifier, \simulator$ such that:
\begin{description}
    \item[CRS generator:] $\kgen$ is a $\ppt$ algorithm that takes an NP-relation $\REL$ as input, runs a one time setup routine, and outputs a \emph{common reference string} (CRS) $\crs$ for this relation along with a \emph{trapdoor} $\td$.
    \item[Prover:]  $\prover$ is a $\ppt$ algorithm that given $(\crs, \inp, \wit)$, such that $(\inp, \wit) \in \REL$, outputs an argument $\pi$.
    \item[Verifier:] $\verifier$ is a $\ppt$ algorithm that on input $(\crs, \inp, \pi)$ returns either $0$ (reject) or $1$ (accept).
    \item[Simulator:] $\simulator$ is a $\ppt$ algorithm that on input $(\crs, \td, \inp)$ outputs an argument $\pi$.
\end{description}

We require a proof system $\snark$ to have the following four properties:
\begin{description}
    \item[Completeness:] $\snark$ is complete if an honest verifier accepts a proof made by an honest prover. That is, the verifier accepts a proof made for $(\inp, \wit) \in \REL$.
    \item[Knowledge soundness:] $\snark$ is knowledge-sound if from an acceptable proof $\pi$ for instance $\inp$ it is feasible for a specialized algorithm called \emph{extractor} to extract a witness $\wit$ such that $(\inp, \wit) \in \REL$.\footnote{note that ``knowledge soundness'' $\Rightarrow$ ``soundness''.}
    \item[Zero knowledge:] $\snark$ is zero-knowledge if for any $\inp \in \LAN$ no adversary can distinguish a proof made by an honest prover on input $(\crs, \inp, \wit)$ from a proof made by the \emph{simulator} on input $(\crs, \inp, \td)$ \emph{but no witness} $\wit$.
    \item[Succinctness:] $\snark$ is succinct if the proof $\pi$ is sub-linear in the size of the instance and witness.
\end{description}

We note that despite \emph{knowledge soundness} being stronger than ``plain'' soundness, it is sometimes still too weak of a notion to satisfy the requirements necessary to deploy a scheme in real world systems. In fact, knowledge-soundness does not protect against Man-in-The-Middle (MiTM) attacks where an adversary can forge an acceptable SNARK after seeing a set of acceptable arguments. In other words, knowledge soundness does not ensure that the SNARK is non-malleable. Recent zk-SNARKs such as~\cite{DBLP:conf/crypto/GrothM17,DBLP:journals/iacr/BoweG18} comply with a strong notion of soundness -- simulation extractability (or simulation knowledge soundness)~\cite{DBLP:conf/focs/Sahai99,DBLP:conf/crypto/SantisCOPS01} -- which prevents MiTM attacks, and which is desirable in many real life applications. SNARKs satisfying the property below are referred to as SE-SNARKs:
\begin{description}
    \item[Simulation extractability:] $\snark$ is simulation-extractable if from any proof $\pi$ for instance $\inp$ output by an adversary with access to an oracle producing simulated proof on given inputs, it is possible for an extractor to extract a witness $\wit$, such that $(\inp, \wit) \in \REL$.
\end{description}

\subsection{Recursive composition of SNARKs}\label{subsection-recursion}

As we know that SNARKs like~\cite{DBLP:conf/eurocrypt/Groth16} can be used to generate arguments for any $\npol$ statements (i.e.~$\inp \in \LAN$, where $\LAN$ is an $\npol$ language), and since we know that the verification algorithm ran by the verifier $\verifier$ is itself in $\npol$, one can wonder if it is possible to generate an argument certifying that another argument has correctly been verified. In other words, is it possible to generate a proof that another proof has correctly been checked?\footnote{In the following we will use ``proof'' and ``argument'' interchangeably to denote computationally sound proofs.}

This question was first studied by Valiant~\cite{DBLP:conf/tcc/Valiant08} who proposed composable proofs of knowledge as a way to achieve \emph{Incrementally Verifiable Computation} (IVC). Further, Bitansky, Canetti, Chiesa and Tromer introduced the notion of \emph{Proof Carrying Data}~\cite{DBLP:conf/innovations/ChiesaT10}, along with a ``boostrapping'' technique~\cite{DBLP:conf/stoc/BitanskyCCT13} to obtain complexity-preserving SNARKs allowing to recursively compose proofs.
The first practical instantiation of recursive (pairing-based) SNARK composition was achieved by Ben-Sasson et al.~\cite{DBLP:journals/algorithmica/Ben-SassonCTV17} using cycles of MNT elliptic curves~\cite{Miyaji2001NewEC} as further studied by Chiesa et al.~\cite{DBLP:journals/siaga/ChiesaCW19}.
Informally, using special tuples of elliptic curves such as cycles (see~\cite{DBLP:journals/em/SilvermanS11,DBLP:journals/siaga/ChiesaCW19} for formal definitions) allows to remove overhead due to finite field characteristic mismatch, and allows to achieve ``infinite recursive SNARK composition''\footnote{it is possible to generate a proof $A$ that another proof $B$ was correctly verified, where $B$ is itself a proof that another proof $C$ was correctly verified, where $C$ is itself a proof that another proof $D$ was verified etc.}

Recent work such as~\cite{DBLP:journals/iacr/BoweCGMMW18} showed that \emph{bounded recursion} was sufficient to build meaningful applications. In their construction, the authors introduced a \emph{pairing-friendly amicable chain} of elliptic curve instantiated as $(\BLSZexe, \CPcurve)$, where $\CPcurve$ is obtained via the Cocks-Pinch method~\cite{CocksPinch} (see~\cite[Section 4.1]{DBLP:journals/joc/FreemanST10}). A more efficient instantiation of the chain was later proposed by El Housni and Guillevic~\cite{DBLP:journals/iacr/HousniG20}.

\subsection{Ethereum}

In the following, we assume familiarity with blockchain systems, and more precisely $\ethereum$. We refer the reader to~\cite{ethereum-whitepaper}, or~\cite{yellow-paper} for an introduction.

\section{Zecale}

\subsection{Motivations}

Despite their broad interest, blockchain systems are often criticized for their inability to ``scale''. Unfortunately, multi-dimensional notions like \emph{scalability} are often simplified to simple metrics such as \emph{Transactions Per Second} (TPS).
Such simplifications turn the ``blockchain scalability problem'' into a TPS maximization problem that fails to capture the true nature of the issue.

\paragraph{What does scalability mean in the context of blockchain systems?}
Answering this question is fundamental to frame the problem that needs to be solved.
Well established payment systems can support thousands of transactions per second, and are systematically used as comparison point to assess the performances of blockchains. Importantly, however, the value proposition of blockchain systems lies in their decentralized nature. As such, keeping the network as decentralized as possible by preserving the network's performances under addition of new nodes in the distributed system is paramount.

Hence, a scalable blockchain system is one that can support a large number of users (high throughput/TPS), as well as a large number of untrusted validator nodes (highly decentralized).
This raises the following paradox:

\medskip

\begin{minipage}[t]{0.95\textwidth}
    \emph{Intuitively, maximizing TPS necessitates to produce more transactions (data) that are processed by the system. However, maximizing the number of validating nodes on the system requires to keep the hardware requirements -- to run a node -- (bandwidth, processing, storage) as low as possible, and thus requires to keep the amount of data, produced and processed, as small as possible.}
\end{minipage}

\medskip

In the following, we refer to a ``scalable blockchain'' as one that implements the right trade-offs allowing to maximize the combination of all scalability parameters.

\paragraph{The cost of privacy.}

Another long standing issue with blockchains is the lack of privacy. In fact, because of their very nature of ``append only'' distributed ledgers, all transactions need to be validated by all miners. As commonly assumed, the transaction data needs to be visible to carry out the verification, which roughly means that \emph{sender}, \emph{recipient} and \emph{amount} of a transaction need to be ``in the clear'' to keep the system sound.
Protocols such as Zerocash/Zcash~\cite{DBLP:conf/sp/Ben-SassonCG0MTV14,zcash-specs}, however, contrast with approaches relying on transparency for security, and rather propose to rely on zk-SNARKs as a way to prove that transactions ``follow the rules of the system'' without disclosing their attributes.

Recently, Rondelet and Zajac~\cite{DBLP:journals/corr/abs-1904-00905} leveraged the (quasi\footnote{due to the need to pay gas for each operations on the $\ethereum$ state, and due to the block gas limit, it is clear that one can easily come up with a state transition that cannot be executed on $\ethereum$ (either because the cost of carrying out the state transition is bigger than the block gas limit, or because the cost of the state transition is bigger than the current supply of $\ether$)}) Turing-complete smart-contract platform of $\ethereum$ as a way to encode a privacy preserving state transition similar to Zerocash. However, and not surprisingly, privacy preserving state transitions like $\zeth$~\cite{DBLP:journals/corr/abs-1904-00905} are significantly more expensive (gas-wise) to carry out than plain $\ethereum$ transactions.\footnote{this is not surprising. Indistinguishability has a price.}
This is partially due to high storage requirements on the smart-contract, and the necessity to carry out multiple cryptographic checks as part of the state transition. Like most SNARK-based applications, one check carried out during the state transition is the verification of the SNARK proof $\pi$ for the base application statement (e.g.~\zeth).

The common framework for SNARK-based layer 2 applications (i.e.~smart-contract) is to:
\begin{enumerate}
    \item Verify the SNARK proof
    \item If the SNARK proof verification is successful, then, execute the state transition specified by the application logic (e.g.~carry out changes in the blockchain state etc.)
\end{enumerate}

Implementing complex smart-contract logic may require to pass a lot of data as input to the smart contract. This means, broadcasting big transactions and blocks on the peer-to-peer network and processing big pieces of data as part of the block mining process. This exacerbates the ``scalability paradox'' above-mentioned.

\paragraph{Scalable privacy.}

$\zecale$ aims to lessen the impact of privacy preserving state transitions like $\zeth$ on the overall system, by minimizing the amount of data sent and processed on-chain. To do this, $\zecale$ off-loads the verification of all zk-SNARK proofs of an application to a piece of software called the \emph{aggregator}. This software component uses recursive proof composition in order to generate a proof of computational integrity certifying that all zk-SNARK proofs have correctly been verified off-chain. This proof is then sent on-chain along with the instances associated with the aggregated proofs. This \emph{unique} proof is then verified on-chain, and each primary inputs are processed according to the associated verification bit.
This technique allows to \emph{aggregate} $N$ SNARK proofs into a single one, and allows to send a \emph{unique} transaction on-chain, without altering the soundness of the system -- \emph{the system remains sound as long as the SNARK-scheme is secure} (see~\cref{subsec:security} for more details). This idea is summarized~\cref{fig:transaction-execution-zecale}

\begin{figure}
    \begin{pchstack}[center]
        \procedure[linenumbering, syntaxhighlight=auto, addkeywords={abort}]{$\includeTxInBlock(\{\zktx_i\}_{i \in [N]})$}{%
            \pcfor i \in [N] \pcdo \\
            \pcind \pcif \neg \snark.\verifier(\zkAppCRS, \zktx_i.\pi, \zktx_i.\inp) \pcthen \\
            \pcind \pcind abort \\
            \pcind \pcelse \\
            \pcind \pcind \execAppLogic(\zktx_i.\inp) \\
            \pcind \pcendif \\
            \pcendfor
        }

        \pchspace

        \procedure[linenumbering, syntaxhighlight=auto, addkeywords={abort}]{$\includeTxInBlock(\aggrtx)$}{%
            \pcif \neg \snark.\verifier(\zecaleCRS, \aggrtx.\pi, \aggrtx.\inp) \pcthen \\
            \pcind abort \\
            \pcendif \\
            \pcfor i \in |\aggrtx.\inputs| \pcdo \\
            \pcind \execAppLogic(\aggrtx.\inputs[i]) \\
            \pcendfor
        }
    \end{pchstack}
    \caption{Left: transaction execution routine without $\zecale$. A set of transactions are processed, all transactions contain a zk-SNARK proof that needs to be verified. Right: transaction execution routine when $\zecale$ is used. Only one transaction -- and SNARK proof -- is processed. $N$ represents the maximum number of $\zktx$ that can be included in a block, where $\zktx$ is a transaction containing a SNARK.}
\label{fig:transaction-execution-zecale}
\end{figure}

\subsection{Technique}

As above-mentioned, $\zecale$ leverages recursive proof composition in order to delegate the individual SNARK proof verification to an off-chain software component -- the \emph{aggregator} (see~\cref{fig:aggregator}). The proof generated by the aggregator is then sent on-chain, along with the primary inputs of all verified proofs in order to execute the state transitions on the state machine (e.g.~\ethereum). Since the proof generated by the aggregator is verified on-chain by all miners, only ``one level of recursion'' is needed in \zecale. This is similar to~\cite{DBLP:journals/iacr/BoweCGMMW18} in that regard.

\begin{figure}[h!]
\centering
\includegraphics[width=1\textwidth]{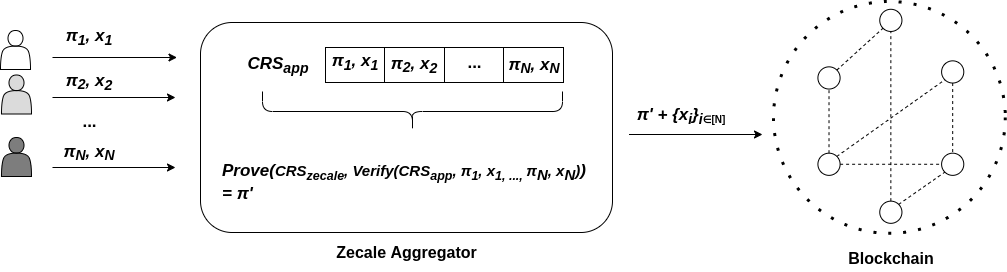}
\caption{Simplified representation of the $\zecale$ aggregator}\label{fig:aggregator}
\end{figure}

As a consequence, we propose to use a \emph{pairing-friendly amicable chain} as a way to implement efficient bounded recursion and bypass several open problems (see~\cite{DBLP:journals/siaga/ChiesaCW19}) and inefficiencies related to the use of low embedding degree cycles of elliptic curves\footnote{if the embedding degrees of the curves are small, then dlog and pairing-based hardness assumptions need to be enforced via the use of prime fields with larger characteristics which renders group and field operations less efficient since more ``limbs'' are needed to represent the algebraic structures' elements.}.

The use of the chain of elliptic curves allows us to generate a set of ``nested'' zk-SNARKs over the first curve, and to generate a ``wrapping'' SNARK\footnote{note that this proof is not zero-knowledge} -- proving correct verification of the set of nested zk-SNARKs -- over the second curve.
As such, it is required that all SNARK-based applications aiming to be used along with $\zecale$ generate zk-SNARK proofs over the first curve of the \emph{pairing-friendly amicable chain} used. The authors of $\zexe$~\cite{DBLP:journals/iacr/BoweCGMMW18} initially proposed a chain instantiated by a curve from the BLS12 family ($\BLSZexe$) and a curve generated using the Cocks-Pinch method. A more efficient instantiation made of $\BLSZexe$ and a Brezing-Weng~\cite{DBLP:journals/dcc/BrezingW05} curve -- denoted $\BWSix$ -- was recently proposed by El Housni and Guillevic~\cite{DBLP:journals/iacr/HousniG20}.
Any such 2-chain (resp.~cycle) of pairing-friendly elliptic curves can be used in $\zecale$. We propose to use $(\BLSZexe, \BWSix)$ as it reflects the state of the art at the time of writing.

\subsection{$\zecale$ language}

In order to provably verify each individual ``nested'' zk-SNARK proofs, it is necessary to encode the SNARK verification algorithm of the base application (i.e.~$\snarkApp.\verifier$) in the correct form: one allowing to efficiently generate a proof of computational integrity of the execution of the algorithm.
This way to encode machine computation so that its correctness can easily be verified via few probabilistic algebraic checks is called \emph{arithmetization}.
SNARKs such as~\cite{DBLP:conf/eurocrypt/Groth16} use a type of arithmetization called \emph{Quadratic Arithmetic Programs} (QAPs) introduced in~\cite{DBLP:conf/eurocrypt/GennaroGP013}.
While describing the full set of algebraic constraints encoding the SNARK verification algorithm as an arithmetic circuit -- defined over a finite field -- is outside of the scope of this paper, we provide below a high-level description of the $\npol$-relation ($\RELZECALE$) characterizing the language of $\zecale$ denoted $\LANZECALE$. This is summarized~\cref{fig:zecale-language}

\begin{figure}[ht]
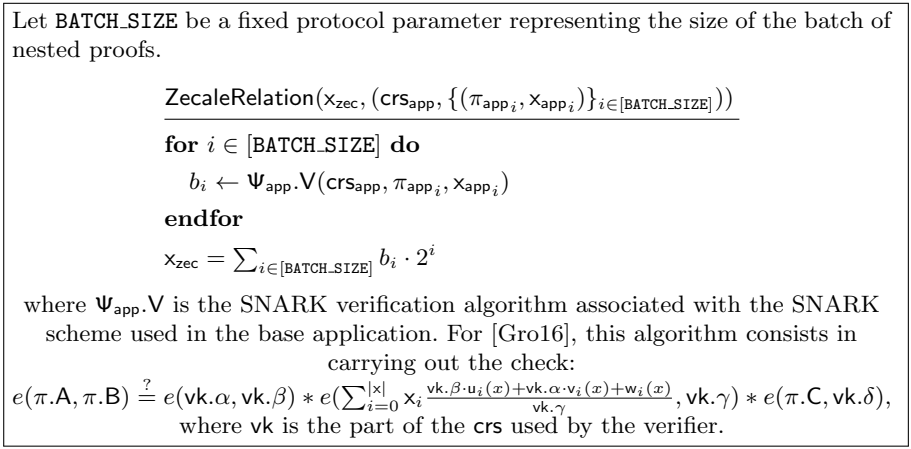

    \fbox{%
        \begin{minipage}[t]{0.95\textwidth}
            Let $\batchSize$ be a fixed protocol parameter representing the size of the batch of nested proofs.
            \medskip

            \centering
            \procedure{$\zecaleRelation(\inpZecale, (\zkAppCRS, \{(\zkpBaseApp_i, \inpBaseApp_i)\}_{i \in [\batchSize]}))$}{%
                \pcfor i \in [\batchSize] \pcdo \\
                \pcind b_i \gets \snarkApp.\verifier(\zkAppCRS, \zkpBaseApp_i, \inpBaseApp_i) \\
                \pcendfor \\
                \inpZecale = \textstyle\sum_{i \in [\batchSize]} b_i \cdot 2^i
            }

            \medskip

            where $\snarkApp.\verifier$ is the SNARK verification algorithm associated with the SNARK scheme used in the base application. For~\cite{DBLP:conf/eurocrypt/Groth16}, this algorithm consists in carrying out the check: $\pair(\pi.\piA, \pi.\piB) \isEq \pair(\vk.\alpha, \vk.\beta) * \pair(\sum_{i = 0}^{|\inp|} \inp_i \frac{\vk.\beta \cdot \polU_i(x) + \vk.\alpha \cdot \polV_i(x) + \polW_i(x)}{\vk.\gamma}, \vk.\gamma) * \pair(\pi.\piC, \vk.\delta)$, where $\vk$ is the part of the $\crs$ used by the verifier.
        \end{minipage}
    }
    \caption{Pseudo-code representation of the set of constraints that all pairs of inputs ($\inp$, $\wit$) need to satisfy to belong to $\RELZECALE$}\label{fig:zecale-language}
\end{figure}

The attentive reader may realize that the verification check of $\groth$ as described in~\cref{fig:zecale-language} is linear with the number of public inputs. As multiple SNARK-based applications may have different number of primary inputs, one needs to overcome some challenges in order to enable the $\zecale$ relation to be generic enough to be used along with a wide class of SNARK-based applications. In fact, it is important to remember that $\groth$ can be used to prove any statement in $\npol$ (due to the $\npol$-completeness of QAPs), but a new setup phase needs to be ran for each new language, posing some practical challenges.

Fortunately, the authors of~\cite{DBLP:conf/eurocrypt/GennaroGP013} proposed a ``trick'' to break the linear complexity of the SNARK verification check by simply applying an ordinary (i.e.~not extractable) collision-resistant hash function to the statement.
We follow the same approach and propose to extend the NP-relations of all base applications --- that aim to be used with $\zecale$ --- with this additional hashing step to reduce the number of their primary inputs to the same constant.


 To simplify notations in the following sections, we define, for each prime $p$ and hash function $\hash$, the following constant $\constant{InpNb} = \ceil{\variable{lh}/(\floor{\log_2(p)} + 1)}$, where $\variable{lh}$ is the digest length of $\hash$. Moreover, we introduce the two following functions: $\toField_{\hash, p} \colon \bin^{\variable{lh}} \to (\FF_p)^{\constant{InpNb}}$ which takes a hash digest and returns its ``prime field representation''\footnote{Converting a bit string to a field element can be done by taking the sum of all $i$th bits of the string multiplied by $2^i$, e.g.~$(110101)_2$ is represented by $(\textstyle{\sum_{i \in [\cardinality{(110)}]}} (110)[i] \cdot 2^i, \textstyle{\sum_{i \in [\cardinality{(101)}]}} (101)[i] \cdot 2^i) = (6, 5) \in (\FF_7)^2$} (the function is injective, i.e.~if $\variable{lh} > \floor{\log_2(p)} + 1$, $\toField_{\hash, p}$ returns the tuple of elements in $\FF_p$ necessary to \emph{uniquely} represent the digest ``in the field''), and $\toDigest_{\hash, p} \colon (\FF_p)^{\constant{InpNb}} \to \bin^{\variable{lh}}$ defined such that for all digests $h$ of $\hash$, we have $\toDigest_{\hash, p}(\toField_{\hash, p}(h)) = h$. Finally, $\FF_\rNested$ (resp.~$\FF_\rWrapping$) is the scalar field of the first (resp.~second) curve of the 2-chain we use.

Using these additional notations, we represent the ``generic'' $\zecale$ relation~\cref{fig:zecale-language-generic}.

\begin{figure}[!htb]
    \fbox{%
        \begin{minipage}[t]{0.95\textwidth}
            Let $\batchSize$ be a fixed protocol parameter representing the size of the batch of nested proofs.
            \medskip

            \centering
            \procedure{$\zecaleRelation(\inpZecale, (\zkAppCRS, \{(\zkpBaseApp_i, \inpBaseApp_i)\}_{i \in [\batchSize]}))$}{%
                \pcfor i \in [\batchSize] \pcdo \\
                \pcind \pclinecomment{As assumed, the relation of the base application uses}\\
                \pcind \pclinecomment{$\hash$ to do the ``hashing trick'' described in [GGPR13].}\\
                \pcind \pclinecomment{Here $\inpBaseApp_i$ represents the ``non-hashed'' instance that}\\
                \pcind \pclinecomment{we now hash to verify the nested SNARK.}\\
                \pcind \inpZecale.\inpPacked_i = \toField_{\hash, \rNested}(\hash(\inpBaseApp_i))\\
                \pcind b_i \gets \snarkApp.\verifier(\zkAppCRS, \zkpBaseApp_i, \inpZecale.\inpPacked_i) \\
                \pcendfor \\
                \inpZecale.\inpValidity = \textstyle\sum_{i \in [\batchSize]} b_i \cdot 2^i \\
                \inpZecale.\vkHash = \toField_{\hash, \rWrapping}(\hash(\zkAppCRS))\\
            }
        \end{minipage}
    }
    \caption{Pseudo-code representation of the \emph{generic} $\zecale$ relation}\label{fig:zecale-language-generic}
\end{figure}

In the rest of the document, we use $\RELZECALE$ to refer to the NP-relation described~\cref{fig:zecale-language-generic}

\paragraph{A note on zero-knowledge.}

Since the $\zecale$ aggregator generates a ``wrapping'' proof to certify the correct verification of a set of ``nested'' proofs, it is not necessary for the ``wrapping'' proof to be zero-knowledge. In fact, if ``nested'' SNARKs already are zero-knowledge, there is no additional security gain to also enforce zero-knowledge for the ``wrapping'' proof. As such, the randomization steps carried out by the prover $\prover$, during the proof generation, to protect against malicious verifiers, and ensure zero-knowledge, can be omitted. This allows to fasten the aggregator algorithm by removing unnecessary operations.

\paragraph{A note on batch verification.}

A well known method to fasten the verification of a set of SNARKs is to use batching techniques. While it may be tempting to modify the NP-relation $\RELZECALE$~\cref{fig:zecale-language-generic} to generate a proof of correct \emph{batch verification} of a set of zk-proofs, this approach presents a set of \emph{practical} limitations. We investigate the case of batch verification in more details in~\cref{appendix:batch-verif} and \cref{appendix:practical-considerations}.

\subsection{Smart-contract logic}

In addition to generate SNARK proofs for statements of the form $\inp \in \LANZECALE$, $\zecale$ requires some smart-contract logic to function. In fact, after aggregating all the base application's zk-SNARKs, it is necessary to verify the ``wrapping'' proof on-chain, and execute the base application logic for all the primary inputs associated with \emph{valid} ``nested'' proofs.

The smart-contract specifying the set of $\zecale$ state transitions uses a \emph{dispatch mechanism} that forwards the set of instances -- associated to \emph{valid} proofs -- to the base application contract that will then process them by executing the base application logic on the state machine.

We provide~\cref{fig:zecale-smart-contract} a pseudo-code description of the $\processAggregatedTx$ function representing the logic of the smart-contract $\ZecaleContract$ specifying $\zecale$ state transitions. Furthermore, we represent~\cref{fig:zecale-smart-contract-dispatch} the difference between a stand-alone base-application and one used along with $\zecale$.


\begin{figure}[!htb]
    \centering
    \fbox{%
    \parbox{1.015\linewidth}{%
    \begin{pchstack}
        \begin{pcvstack}
        \procedure[linenumbering, syntaxhighlight=auto, addkeywords={abort}]{$\baseAppContract.\constructor(\zkAppCRS, \ldots)$}{%
            \pclinecomment{1. Write the CRS in the storage}\\
            \storage[\fieldAppCRS] \gets \zkAppCRS\\
            \pclinecomment{2. Execute further}\\
            \pclinecomment{application-specific logic}\\
            \ldots \\
            \pcreturn\\
            \\
        }

        \pcvspace

        \procedure[linenumbering, syntaxhighlight=auto, addkeywords={abort}]{$\baseAppContract.\processTx(\zkpBaseApp, \inpBaseApp)$}{%
            \pclinecomment{1. Check that the primary inputs}\\
            \pclinecomment{are in the proper scalar field}\\
            \pcif \inpBaseApp \not\in (\FF_\rNested)^{\cardinality{\inpBaseApp}}\\
            \pcind abort\\
            \pcendif\\
            \pclinecomment{2. Verify the SNARK}\\
            \crs \gets \storage[\fieldAppCRS]\\
            \pcif \neg \snarkApp.\verifier(\crs, \zkpBaseApp, \inpBaseApp):\\
            \pcind abort\\
            \pcendif\\
            \pclinecomment{3. Execute the application logic}\\
            \execAppLogic(\inpBaseApp)\\
            \pcreturn 1
        }
        \end{pcvstack}


        \begin{pcvstack}
        \procedure[linenumbering, syntaxhighlight=auto, addkeywords={abort}]{$\modifiedBaseAppContract.\constructor(\zkAppCRS, \zecaleContractAddress, \ldots)$}{%
            \pclinecomment{1. Write the hash of the CRS in the storage}\\
            \storage[\fieldAppCRS] \gets \hash(\zkAppCRS)\\
            \pclinecomment{2. Write the address of the $\zecale$ contract }\\
            \pclinecomment{in the storage}\\
            \storage[\fieldZecaleAddr] \gets \zecaleContractAddress\\
            \pclinecomment{3. Execute further application-specific logic}\\
            \ldots \\
            \pcreturn
        }

        \pcvspace

        \procedure[linenumbering, syntaxhighlight=auto, addkeywords={abort}]{$\modifiedBaseAppContract.\dispatch(\vkHash, \byteData)$}{%
            \pclinecomment{1. Check that the caller is the registered}\\
            \pclinecomment{$\zecale$ contract}\\
            \pcif \neg\ \msgSender = \storage[\fieldZecaleAddr]\\
            \pcind abort\\
            \pcendif\\
            \pclinecomment{2. Check that the received inputs are for}\\
            \pclinecomment{the right language}\\
            \pcif \neg\ \toDigest_{\hash, \rWrapping}(\vkHash) = \storage[\fieldAppCRS]\\
            \pcind abort\\
            \pcendif\\
            \decodedData \gets \decodeBytes(\byteData) \\
            \pclinecomment{3. Check that the primary inputs are in}\\
            \pclinecomment{the proper scalar field}\\
            \pcif \decodedData \not\in (\FF_\rNested)^{\cardinality{\decodedData}}\\
            \pcind abort\\
            \pcendif\\
            \pclinecomment{4. Execute the base application logic}\\
            \pcfor i \in [|\decodedData|] \pcdo\\
            \pcind \execAppLogic(\decodedData[i])\\
            \pcendfor\\
            \pcreturn 1
        }
        \end{pcvstack}
    \end{pchstack}
    } 
    } 
    \caption{Left: Pseudo-code of the state transition of a stand-alone SNARK-based application. Right: Pseudo-code of the state transition of a SNARK-based application implementing the ``dispatch mechanism'' to be used with $\zecale$.}
    \label{fig:zecale-smart-contract-dispatch}
\end{figure}

\clearpage

\begin{figure}[!htb]
    \centering
    \fbox{%
        \procedure[linenumbering, syntaxhighlight=auto, addkeywords={abort}]{$\ZecaleContract.\processAggregatedTx(\zkpZecale, \inpZecale, \{\inpBaseApp_i\}_{i \in [\batchSize]}, \ZappContractAddress)$}{%
            \pclinecomment{1. Check the inputs before carrying out any expensive computation}\\
            \pcif \inpZecale.\inpValidity\ >\ \textstyle\sum_{i \in [\batchSize]} 2^i \pcthen\\
            \pcind abort\\
            \pcendif\\
            \pclinecomment{Make sure that the application inputs have not been maliciously}\\
            \pclinecomment{tampered with by the $\zecale$ aggregator}\\
            \pcfor i \in [\batchSize] \pcdo \\
            \pcind \pcif \neg \left[\inpZecale.\inpPacked_i = \toField_{\hash, \rNested}(\hash(\inpBaseApp_i))\right] \pcthen\\
            \pcind \pcind abort \\
            \pcind \pcendif \\
            \pcendfor \\
            \pclinecomment{2. Check the aggregation/wrapping SNARK}\\
            \pcif \neg \snarkZecale.\verifier(\zecaleCRS, \zkpZecale, \inpZecale) = \true \pcthen:\\
            \pcind abort \\
            \pcendif \\
            \pclinecomment{3. Execute the base application state transitions}\\
            \pclinecomment{3.1 If none of the nested proofs are correct, there is nothing to dispatch}\\
            \pcif \inpZecale.\inpValidity\ = 0\\
            \pcind abort\\
            \pcendif\\
            \pclinecomment{3.2 Otherwise, dispatch the instances associated with the valid nested proofs}\\
            \dispatchData \gets \{0\}^\batchSize \\
            \pcfor i \in [\batchSize] \pcdo \\
            \pcind \pcif \inpZecale.\inpValidity\ \&\ \mathtt{0x1} \pcthen\\
            \pcind \pcind \dispatchData[i] \gets \inpBaseApp_i \\
            \pcind \pcendif \\
            \pcind \pclinecomment{Right shift by 1 position}\\
            \pcind \shr(\inpZecale.\inpValidity, 1) \\
            \pcendfor \\
            \modifiedBaseAppContract \gets \createContractInstance(\ZappContractAddress)\\
            \byteData \gets \encodeToBytes(\dispatchData) \\
            \modifiedBaseAppContract.\dispatch(\inpZecale.\vkHash, \byteData)\\
            \pcreturn 1
        }
    }
    \caption{Pseudo-code specifying the $\zecale$ state transition}
    \label{fig:zecale-smart-contract}
\end{figure}

In brief, the $\zecale$ state transition can be decomposed into 3 steps:
\begin{description}
    \item[Check the inputs:] As nested proofs are verified off-chain, it is important to make sure that the instances of the nested proofs, forwarded on-chain by the aggregator to execute the base application logic, are the same as the one used during the off-chain ``nested'' proof verification. In other words, it is key to make sure that the base application state transition processes the \emph{``right instances''} (instances that haven't been tampered with by the aggregator).\footnote{note that further security checks need to be enforced at the base application layer to prevent front-running and malleability issues in the context of a malicious aggregator (see~\cref{subsec:aggr-market} for more context, and~\cite{zeth-specs} for an example of non-malleable construction)}. Likewise, we check that $\inpZecale.\inpValidity \in [\textstyle\sum_{i \in [\batchSize]} 2^i]$ in order to abort the state transition as soon as possible if not the case.
    \item[Verify the aggregation SNARK proof:] Once the inputs are checked to be of the right form, the wrapping proof is verified. If the proof verifies correctly, then network participants have overwhelming confidence that the set of nested proofs have correctly been verified off-chain.
    \item[Forward the application data to the base application:] The base application verifies that:
    \begin{itemize}
        \item The calling contract is the genuine $\zecale$ contract --- to be sure that all necessary checks have successfully been carried out.
        \item The received data is made of instances for this application.
    \end{itemize}
    Finally, the base application logic is executed. Note however, that we have explicitly represented the check that consists in verifying that all instances ($\inpBaseApp$) associated to valid nested SNARKs lie in the right finite field. In fact, failure to do such a check can expose the underlying application to modular arithmetic-based attacks like~\cite{poma-attack}.
\end{description}

\begin{example}
    Let $\batchSize = 3$, $\inpZecale.\inpValidity = 5$, and let $(\inpBaseApp_0, \inpBaseApp_1, \inpBaseApp_2)$ be the application instances.
    Since the size of the proofs batch is $3$, at most $3$ zk-SNARKs are valid. As such $\inpZecale.\inpValidity$ is bounded by $(111)_2 = 7$. We check that $\inpZecale.\inpValidity < 7$, which is satisfied here.
    The binary representation of $\inpZecale.\inpValidity$ is $(101)_2$. As such, the nested zk-SNARKs at indices $0$ and $2$ in the batch verified correctly, while the proof at index $1$ was deemed incorrect by the verification algorithm. This means that, out of the tuple of instances $(\inpBaseApp_0, \inpBaseApp_1, \inpBaseApp_2)$, only $\dispatchData \gets (\inpBaseApp_0, \inpBaseApp_2)$ will be forwarded to the base application contract to be processed (e.g.~added to the Merkle tree of commitments etc.~if the base application is $\zeth$).
\end{example}

\subsection{Instantiation of the SNARK scheme}

Despite the fact that $\groth$ is not universal~\cite{DBLP:conf/crypto/GrothKMMM18} and only weakly simulation extractable~\cite{cryptoeprint:2020:811}, it remains of broad interest in applied settings because of its nearly optimal argument size and efficiency.

As our main focus is to minimize the size of data exchanged and processed by miners while keeping the cost of the $\zecale$ state transition as small as possible, we believe that $\groth$ is a good candidate to instantiate the SNARK scheme used in $\zecale$.

It is important to note, however, that recent proof systems such as, e.g.~\cite{DBLP:journals/iacr/GabizonWC19,DBLP:conf/eurocrypt/ChiesaHMMVW20} could also be used to instantiate the SNARK scheme used in $\zecale$. Likewise, there are no requirements to use the same proof system in $\zecale$ and in the base application. For instance, it is feasible for $\zeth$ to use $\groth$, while generating $\zecale$ wrapping proofs can be done using another proof system, such as~\cite{DBLP:conf/crypto/GrothM17} for instance.

\subsection{Security}\label{subsec:security}

In the following, we prove that $\zecale{}$ is secure. Namely, we show that the protocol preserves the soundness of the underlying blockchain system.

To study the soundness of $\zecale$ we want to show that the probability that an adversary $\adv$ can use $\zecale$ in order to break the soundness of a ledger $\ledger$ is negligible in $\secpar$. We do so by defining the soundness game $\zclSND$ below, and argue that the probability of winning this game (also referred to as $\advantage{\zclSND}{\adv, \ledger}$) is negligible.

\newcommand{\APPS}{\variable{APPS}} 
\newcommand{\capp}{\variable{capp}} 
\newcommand{\myChi}{\raisebox{2pt}{$\chi$}} 

\newcommand{\genNestedProof}{\algostyle{genNestedProof}}
\newcommand{\verifNestedProof}{\algostyle{verifNestedProof}}
\newcommand{\genWrappingProof}{\algostyle{genWrappingProof}}
\newcommand{\verifWrappingProof}{\algostyle{verifWrappingProof}}

We denote by $\APPS$ a set of applications deployed on $\ledger$ during the execution of the $\setup$ algorithm. In the following, an application $\variable{app}$ will be represented by a tuple $(\zkAppCRS, \appContractAddress, \ZappContractAddress)$. Moreover, we denote by $\mathcal{I} = \Pi \diamond \myChi,\ \cardinality{\mathcal{I}} = \cardinality{\myChi} = \cardinality{\Pi} = \batchSize$ the set of pairs of SNARKs proofs generated with $\snarkApp$ ($\Pi$), and associated set of primary inputs ($\myChi$). We denote by $\diamond$ the operator that takes two ordered sets of same cardinality as input and builds a resulting set which $i$th element is the pair of the $i$th elements of the input sets.
Furthermore, $\mathcal{S} \subset \myChi$.

\begin{center}
\fbox{%
    \procedure{$\zclSND(\secpar)$}{%
        (\zecaleCRS, \ledger, \APPS) \gets \setup(\secpar)\\
        \state \gets \adv^{\oracle{\ledger, \snarkApp, \snarkZecale}}(\zecaleCRS, \APPS)\\
        \capp \sample \APPS\\
        (\zkpZecale, \inpZecale, \mathcal{I}, \mathcal{S}) \gets \adv^{\oracle{\ledger, \snarkApp, \snarkZecale}}(\zecaleCRS, \APPS, \state, \capp)\\
        b \gets \zecaleP(\zecaleCRS, \zkpZecale, \inpZecale, \myChi, \mathcal{S}, \capp.\ZappContractAddress)\ \land\ \\
        \pcind \left[\exists \inp_i \in \mathcal{S},\ (\pi_i, \inp_i) \in \mathcal{I},\ \neg \appP(\pi_i, \inp_i, \capp.\appContractAddress)\right]\\
        \pcreturn b
    }
}
\end{center}

In the game above, $\adv$ can do 2 types of oracle queries to the ledger $\ledger$:
\begin{itemize}
    \item $\processTx$ which takes $(\pi, \inp, \appContractAddress)$, and which execute the function $\processTx$ of $\baseAppContract$ at address $\appContractAddress$ on input $(\pi, \inp)$.
    \item $\processAggregatedTx$ which takes $(\zkpZecale, \inpZecale, \{\inp_i\}_{i \in [\batchSize]}, \ZappContractAddress)$ as inputs and executes the function $\processAggregatedTx$ of the contract $\ZecaleContract$ on the same inputs.
\end{itemize}

Likewise $\adv$ is allowed to do oracle queries to generate ``nested'' (resp.~``wrapping'') proofs --- i.e.~call $\snarkApp.\prover$ (resp.~$\snarkZecale.\prover$) --- and verify them --- i.e.~execute $\snarkApp.\verifier$ (resp.~$\snarkZecale.\verifier$):
\begin{itemize}
    \item $\genNestedProof$: takes $(\crs, \inp, \wit)$ as input, where $\crs$ is the CRS of one of the applications in $\APPS$, and returns the output of $\snarkApp.\prover(\crs, \inp, \wit)$.
    \item $\verifNestedProof$: takes $(\crs, \pi, \inp)$ as input, where $\crs$ is the CRS of one of the applications in $\APPS$, and returns the output of $\snarkApp.\verifier(\crs, \pi, \inp)$.
    \item $\genWrappingProof$: takes $(\inpZecale, \crs, \{\pi_i, \inp_i\}_{i \in [\batchSize]})$ as input, where $\crs$ is the CRS of one of the applications in $\APPS$, and returns the output of $\snarkZecale.\prover(\zecaleCRS, \inpZecale, (\crs, \{\pi_i, \inp_i\}_{i \in [\batchSize]}))$.
    \item $\verifWrappingProof$: takes $(\zecaleCRS, \zkpZecale, \inpZecale)$ as input, and returns the output of $\snarkZecale.\verifier(\zecaleCRS, \zkpZecale, \inpZecale)$.
\end{itemize}

In the game $\zclSND$, $\zecaleP(\zecaleCRS, \zkpZecale, \inpZecale, \myChi, \mathcal{S}, \ZappContractAddress)$ is the predicate that returns $\true$ if the value of $\decodedData$ (see line 18 in the $\dispatch$ function~\cref{fig:zecale-smart-contract-dispatch}) equals $\mathcal{S}$ when $\ZecaleContract.\processAggregatedTx$ is called on $(\zkpZecale,\allowbreak \inpZecale,\allowbreak \myChi,\allowbreak \ZappContractAddress)$, and returns $\false$ otherwise.

Moreover, $\appP(\pi_i, \inp_i, \appContractAddress)$ is the predicate that returns $\true$ if $\baseAppContract.\processTx(\pi_i, \inp_i)$, where $\baseAppContract$ is the application contract at address $\appContractAddress$, returns $1$. This predicate returns $\false$ in all other cases.

All in all, winning the game above means that the adversary used $\zecale$ as a way to carry out changes in the ledger $\ledger$'s state that should not have been done by only using the stand-alone base application contract (i.e.~without $\zecale$).

\begin{theorem}
    Let $\adv$ be a $\ppt$ adversary. If $\hash$ is a collision-resistant hash function (i.e.~$\advantage{\collRes}{\adv, \hash} \leq f(\secpar)$, $f$ negligible)\footnote{Note that the requirement on $\toField$ being an \emph{injective map} is important, as it allows to be sure that the function is not a ``source of collisions''.}, and if $\snarkZecale$ is a sound SNARK scheme (i.e.~$\advantage{\snarkSND}{\adv, \snarkZecale} \leq g(\secpar)$, $g$ negligible), then $\advantage{\zclSND}{\adv, \ledger} \leq \negl$.
\end{theorem}

\begin{proof}
    Let $X = (\zecaleCRS, \zkpZecale, \inpZecale, \myChi, \mathcal{S}, \capp.\ZappContractAddress)$ be such that $\zecaleP(X) = \true$.
    Additionally, let $\mathcal{I}, \mathcal{S} \subset \myChi$ be such that $\exists \inp_i \in \mathcal{S},\ \allowbreak (\pi_i, \inp_i) \in \mathcal{I},\ \allowbreak \neg \appP(\pi_i,\allowbreak \inp_i,\allowbreak \capp.\appContractAddress)$.

    By looking at $\baseAppContract$ (see~\cref{fig:zecale-smart-contract-dispatch}), we see that for $\appP(\pi_i,\allowbreak \inp_i,\allowbreak \capp.\appContractAddress)$ to return $\false$, at least one of the two predicates below needs to be \emph{not} satisfied:
    \begin{enumerate}[label=(\subscript{E}{{\arabic*}})]
        \item $\inpBaseApp_i \in (\FF_\rNested)^{\cardinality{\inpBaseApp_i}}$
        \item $\snarkApp.\verifier(\zkAppCRS, \pi_i, \inp_i)$
    \end{enumerate}

    \medskip

    We study the probability of these two events below.
    \begin{enumerate}[label=(\subscript{E}{{\arabic*}})]
        \item Since the $\zecaleP$ predicate was satisfied (by assumption), this means that the check $\decodedData \in (\FF_\rNested)^{\cardinality{\decodedData}}$ was satisfied. The probability to satisfy this check on $\zecale$ while not satisfying it on the base application contract is 0 (the same check is done on both $\modifiedBaseAppContract$ and $\baseAppContract$).
        \item We distinguish three cases for which $\zecaleP$ can return $\true$, yet one of the $\appP$ predicates returns $\false$ because the SNARK verification check $\snarkApp.\verifier(\zkAppCRS, \pi_i, \inp_i)$ is not satisfied:
            \begin{enumerate}[label=\Alph*.]
                \item $\zkpZecale \gets \snarkZecale.\prover(\zecaleCRS, \inpZecale, \zkAppCRS,\mathcal{I})$, where $\mathcal{I}$ is a set containing proof/instance pairs for $\capp$, but where \emph{at least one proof ($\pi_i$) in the set is not valid (i.e.~$\inp_i \not\in \LAN^{\project{capp}}$), and its associated instance is in $\mathcal{S}$}.
                For $\zecaleP$ to be $\true$, and since the invalid instance was included in $\mathcal{S}$, this means that the associated invalid proof was deemed valid since the check on $\inpValidity$ in $\RELZECALE$ was successful. Since the proof was correctly rejected on $\appP$, this means that the SNARK $\snarkZecale$ is not sound since the statement $\inpZecale \in \LANZECALE$ was deemed correct while in reality $\inpZecale \not\in \LANZECALE$. We denote the probability of this event by $\advantage{\snarkSND}{\adv, \snarkZecale}$.
                \item $\zkpZecale \gets \snarkZecale.\prover(\zecaleCRS, \inpZecale, \zkOtherAppCRS, \mathcal{I})$, where $\mathcal{I}$ is a set containing only valid proof/instance pairs for $\variable{\widetilde{app}}$, where \emph{$\variable{\widetilde{app}} \neq \capp$}. In this case, it is also clear that $\appP$ will return $\false$ on all input pairs in $\mathcal{I}$ as they have been generated for another application (under another CRS which is different from $\capp.\zkAppCRS$, the one used in $\appP$ in the game). However, for $\zecaleP$ to be true, this means that the check $\toDigest_{\hash, \rWrapping}(\vkHash) = \storage[\fieldAppCRS]$ (line 8 in the $\dispatch$ function in~\cref{fig:zecale-smart-contract-dispatch}) was successful. However, by looking at~\cref{fig:zecale-language-generic}, we know that the value of $\inpZecale.\vkHash$ is constrained to equal to $\toField_{\hash, \rWrapping}(\hash(\variable{\widetilde{app}}))$ for $\inpZecale$ to be in $\LANZECALE$. Hence, assuming that $\snarkZecale$ is sound, for $\zecaleP$ to accept, it is necessary to have $\toDigest_{\hash, \rWrapping}(\toField_{\hash, \rWrapping}(\hash(\variable{\widetilde{app}}))) = \hash(\capp)$, which by the definition of $\toDigest_{\hash, \rWrapping}$ and $\toField_{\hash, \rWrapping}$ means that $\hash(\variable{\widetilde{app}}) = \hash(\capp)$, where $\capp \neq \variable{\widetilde{app}}$. This implies that $\adv$ needs to find a collision in $\hash$. We denote the probability of this event by $\advantage{\collRes}{\adv, \hash}$.
                \item $\zkpZecale \gets \snarkZecale.\prover(\zecaleCRS, \inpZecale, \zkAppCRS, \{\pi_i, \inp_i\}_{i \in [\batchSize]})$, where $\{\pi_i, \inp_i\}_{i \in [\batchSize]}$ is a set containing only valid proof/instance pairs for $\capp$. Furthermore, we set $\mathcal{I} \gets \{\pi_i, \widetilde{\inp_i}\}_{i \in [\batchSize]}$, where $\exists j \in [\batchSize],\ \widetilde{\inp_j} \neq \inp_j$ (i.e.~one of the ``nested'' instance has been tampered with by the aggregator after generating $\zkpZecale$). For $\zecaleP$ to accept such input (which is rejected by $\appP$), the check $\inpZecale.\inpPacked_j = \toField_{\hash, \rNested}(\hash(\widetilde{\inp_j}))$ (lines 7-11~\cref{fig:zecale-smart-contract}) needs to be satisfied. However, since $\inpZecale.\inpPacked_j$ is constrained to be equal to $\toField_{\hash, \rNested}(\hash(\inp_j))$ (see~\cref{fig:zecale-language-generic}), since $\toField_{\hash, \rNested}$ is injective and since $\widetilde{\inp_j} \neq \inp_j$, then $\adv$ broke the collision resistance of $\hash$.
            \end{enumerate}
    \end{enumerate}

    All in all, $\advantage{\zclSND}{\adv, \ledger} \leq \prob{E_1} + \prob{E_2}$, where $\prob{E_2} \leq \advantage{\snarkSND}{\adv, \snarkZecale} + 2 \cdot \advantage{\collRes}{\adv, \hash}$. As such, we have $\advantage{\zclSND}{\adv, \ledger} \leq 0 + \left(\advantage{\snarkSND}{\adv, \snarkZecale} + 2 \cdot \advantage{\collRes}{\adv, \hash}\right)$, which is negligible in $\secpar$.

    \qed

\end{proof}

\begin{remark}
    We note that the hash function $\hash$ used to hash the instances associated to the ``nested proofs'' does not necessarily need to be the same hash function as the one used the hash the verification key. It is perfectly feasible --- and (maybe) in some contexts even desirable --- to instantiate these hash functions differently as long as both functions comply with the security requirements mentioned above.
\end{remark}

\begin{remark}
    As observed by Duncan Tebbs~\cite{duncan-attack} it is necessary to check on $\modifiedBaseAppContract$ that the calling $\zecale$ contract (i.e.~$\ZecaleContract$) is genuine in order to be sure, on the application contract, that all the checks necessary to make the protocol sound are properly carried out. Failure to check the address of the calling contract would render the protocol vulnerable to a range of malicious calling contracts.
\end{remark}

\section{Applications}

In this section we present a few types of applications of $\zecale$, and conclude by arguing that $\zecale$ and privacy preserving solutions like $\zeth$ can be used as building blocks to implement blockchain-based digital cash systems.

We note that none of the applications discussed below are mutually exclusive nor that they comprehensively represent the landscape of potential applications of $\zecale$.

\subsection{Application types 1: Blockchain users run the $\zecale$ aggregator locally}

In this setting, the $\zecale$ aggregator piece of software is ran on the network users' machines in order to aggregate their own batches of transactions locally. Doing so allows users to save gas since only one transaction is sent on-chain, and only one proof needs to be verified as part of the smart-contract execution. Running $\zecale$ locally to batch one's transactions allows to save at most $\gasSaved$ gas:
\begin{equation}\label{eq:gas-saved}
    \gasSaved = \txDefaultGas \cdot (\batchSize - 1) + \batchSize \cdot \verifNProofGas - \verifWProofGas\,
\end{equation}
where $\txDefaultGas$ is the intrinsic gas of a transaction, $\verifNProofGas$ the gas necessary to verify one nested proof, and $\verifWProofGas$ is the gas necessary to verify the wrapping proof.

\paragraph{Would batch verification make sense in this setting?}\label{par:batch-verif-app-type1}

In the context where $\zecale$ only aims to be used in this setting (i.e.~as a way to equip network users with a way to batch their transactions locally), it may be tempting to modify $\RELZECALE$ in order to support batch verification of proofs.
In fact, many of the practical security issues discussed in~\cref{appendix:practical-considerations} do not apply here anymore since all proofs in the batch are generated by the user directly.
However, since all nested proofs are both generated \emph{and} batch-verified (as part of the wrapping proof generation) by the user, it is important, for the security of the system, to make sure that a malicious user cannot craft a set of scalars in the batch verification equation~\cref{eq:groth-batch-verify} that could violate the soundness of the system. In other words, it is necessary to make sure that the set of scalars used in the batch verification equation are not under the prover's control. While picking a set of random field elements cannot be enforced using arithmetic gates, one may want to leverage Pseudo-Random Functions (PRFs) as a way to deterministically derive a set of random scalars from the set of proofs in the batch and add this derivation process to $\RELZECALE$. Further discussion on that is provided in~\cref{appendix:batch-verif}.

\subsection{Application type 2: Blockchain miners run the $\zecale$ aggregator}

Another application of $\zecale$ consists in adding the aggregator software logic as part of the block production on the blockchain.
Doing so requires to make some changes to the blockchain protocol -- which in many cases, may not be desired -- but could however, lead to interesting scenarios and platform economics in which miners are incentivized to aggregate transactions in the blocks they produce.
One such protocol modification may, for instance, consist in extending the block production reward with an aggregation reward, and design a penalty that would diminish the block production reward in the event where a block proposed by a miner contains transactions that could have been aggregated.
More drastically, one may want to change the validity rules of a block~\cite[Section 4.3]{yellow-paper}, by enforcing that all blocks containing data that could have been aggregated via $\zecale$, are deemed invalid on the network and thus rejected.

Such protocols would minimize data redundancy on-chain and would allow to minimize the growth of data on the blockchain, at the cost of modifying the base protocol.

\subsection{Application type 3: Creation of an aggregation market}\label{subsec:aggr-market}

Finally, in this $3$rd category of applications we argue that solutions such as $\zecale$ could be used to enrich the blockchain ecosystem by creating new economic opportunities.

Privacy preserving solutions like $\zeth$ allow to keep the value and the recipient of a payment secret. Nevertheless, because of the need to pay for the gas of the state transition, network users can observe when another user sends a $\zeth$ transaction (see~\cite{DBLP:journals/corr/abs-1904-00905} for more discussions on this).

Despite this information leakage (i.e.~$\ethereum$ balances going ``up and down''), the transaction graph remains blurred (see~\cite[Figure 2]{DBLP:journals/corr/abs-1904-00905}). Users emitting $\zeth$ transactions cannot be associated to specific payments. Likewise, while \emph{transaction unobservability} is not achievable on a blockchain system (miners and other network users can see the transactions in the system), we note that the possibility to emit ``dummy $\zeth$ payments'' -- at the cost of paying the price of the state transition -- allows to create additional noise in the system in order to achieve \emph{payment unobservability}.

Nevertheless, this ability to detect when a user triggers a privacy preserving state transition on the distributed state machine (whether it is a dummy payment or not) may prevent the deployment and adoption of solutions like $\zeth$ in countries where the use of Privacy Enhancing Technologies (PETs) is prohibited. Making sure to remove this information leakage is of tremendous importance for the wide deployment of such solutions.

Below, we show how one could leverage solutions such as $\zecale$ along with the use of Anonymous Communication (AC) protocols like mix-networks in order to \emph{bypass} gas-related information leakages and hide the sender of transactions.

\subsubsection{Anonymous communication protocols and sender anonymity.}

As above-mentioned, solutions like $\zeth$ allow to achieve \emph{recipient anonymity} and \emph{relationship anonymity} (as defined in~\cite{Pfitzmann09aterminology,DBLP:journals/jpc/BackesKMMM16}). Unfortunately, \emph{sender anonymity} is not ensured by the protocol because of the need to pay gas to execute state transition on the distributed state machine.

While removing the need to pay gas for each transaction on the blockchain would undoubtedly remove some types of leakages and could ultimately be used to allow transactions to be emitted from newly created accounts (with no funds) as a way to gain sender anonymity; it is clear that such solution would expose the system to a wide class of attacks -- such a Denial of Services (DoS).
Instead, another way to obtain sender anonymity would be to leverage a network of \emph{relay nodes} that would listen to incoming messages/transactions and emit/relay them on-chain. Such service could be rewarded by a \emph{relay fee}.

Resorting to relay nodes as a way to gain sender anonymity by moving the need to pay gas onto another entity is not without risk however. In addition to the need to design a sound crypto-economic protocol and implementing the right incentive structure to reward relays and keep the system secure; it is also necessary to make sure that the \emph{system remains censorship resistant} to make sure that no malicious relay node can prevent transactions from being sent and executed on-chain -- even if deemed irrational by the incentive structure.

Interestingly, if a SNARK-based state transition \emph{only} processes a zero-knowledge SNARK proof along with the associated set of public inputs (that do not leak the sender of the transaction) during its execution, it is possible to use anonymous communication protocols to simply route the zero-knowledge proof and the instance to a relay. Routing such information using an anonymous network based on onion routing~\cite{DBLP:journals/cacm/GoldschlagRS99} (e.g.~\tor~\cite{tor-design}) or mix-networks~\cite{DBLP:series/ais/Chaum03,ba-phd-thesis,DBLP:journals/pieee/SampigethayaP06} along with cryptographic packet format~\cite{DBLP:conf/sp/DanezisG09} (e.g.~\loopix~\cite{DBLP:conf/uss/PiotrowskaHEMD17} or \nym~\cite{nym-lightpaper}) would protect the transaction originator from any malicious relays, and would render any targeted censorship strategy inefficient (see~\cref{fig:relay-shuffle}).

\begin{figure}[!htb]
\centering
\includegraphics[width=1\textwidth]{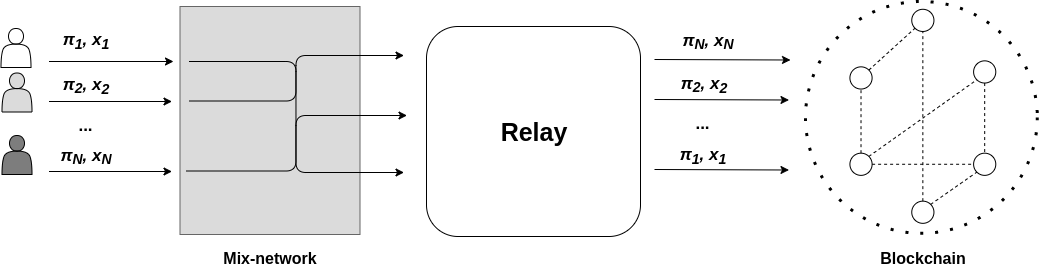}
\caption{Simplified representation of an anonymity layer allowing to send proofs anonymously to a relay}\label{fig:relay-shuffle}
\end{figure}

\subsubsection{$\zecale$ as an aggregation proxy to scale SNARK-based blockchain applications.}

In addition to simply relay received zk-SNARK proofs on-chain, we envision scenarios where \emph{relay nodes} could also run the $\zecale$ aggregator as a way to compress the set of received proofs, into a single SNARK-proof to be sent on-chain (see~\cref{fig:aggregator-shuffle}).

While relaying transactions allows transaction originators to become anonymous, aggregating received zk-SNARKs before relaying them on-chain allows blockchain validators/miners to exchange and process less data. Like ``honest relaying'', ``honest aggregation'' of transactions should also be worth a fee that could, for instance, be paid by blockchain miners to reward the aggregators for their ``data compression'' work, which ultimately renders the blockchain system more efficient.

\begin{figure}[!htb]
\centering
\includegraphics[width=1\textwidth]{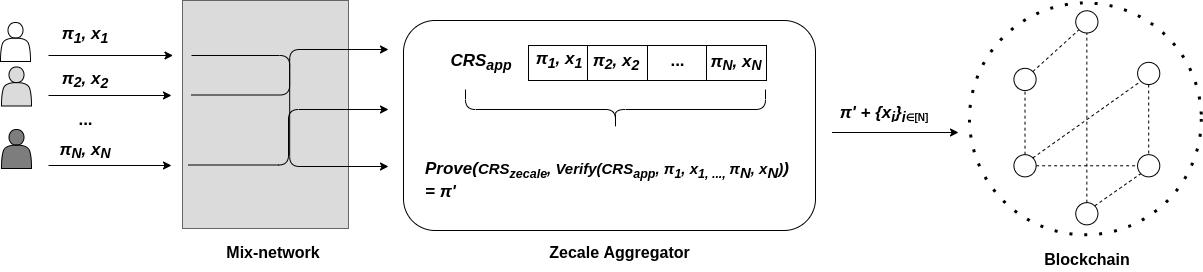}
\caption{Simplified representation of an anonymity layer allowing to send proofs anonymously to a relay running the $\zecale$ aggregator}\label{fig:aggregator-shuffle}
\end{figure}

\subsubsection{Aggregation market.}

While the design of specific incentive structures is deferred for later work, we highlight that delegating the proof aggregation to a set of relays on an \emph{``aggregation network''} can pave the way for what we believe to be interesting scenarios, where all relay/aggregators compete to process as many zk-SNARK proof as possible in order to maximize their fees. Such competition may be a driver for innovation as each aggregator operator would be incentivized to implement state of the art SNARK proof generation algorithms running on efficient hardware (e.g.~GPU) in the hope to capture as much traffic as possible on the aggregation network and generate profit.

\begin{remark}
    We note that for blockchains based on Nakamoto-style consensus~\cite{bitcoin-whitepaper}, block production work is intrinsically different from zk-SNARK proof production.
    As such, we believe that the distinct roles of miners and aggregators are very complementary and reflect the different nature of the computational tasks. Delegating the wrapping proof generation (i.e.~the aggregation of zk-SNARK proofs) to a set of aggregators seems to be a very natural model.
\end{remark}

\paragraph{Toward the formation of ``aggregation pools''?}
As we know that carrying out $\groth$ SNARK verification as part of proving that $\inp \in \LANZECALE$ is expensive (i.e.~pairings are expressed by large R1CS), we believe that new form of ``pools'' (similar to mining pools) could emerge on the aggregation network. In fact, Wu et al.~\cite{DBLP:conf/uss/WuZCPS18} proposed $\dizk$, a system that distributes the generation of a zero knowledge proof across machines in a compute cluster. Some blockchain systems leveraging $\zecale$ may be willing to maximize the size of the batch of aggregated zk-SNARK proofs -- i.e.~maximize $\batchSize$ -- in order to maximize data compression and flatten the growth of the blockchain data. However, doing so would significantly complicate the set of algebraic constraints representing $\RELZECALE$ and increase the degrees of the interpolated polynomials manipulated during the proof generation process. This could be compensated by the formation of ``aggregation pools'' where a set of contributors could allocate some compute resources to the generation of a wrapping proof, and be rewarded pro-rata the resources allocated. We believe that further investigation of the formation of aggregation pools is of great interest. However, we defer such study for later work.

\section{$\zeth$ and $\zecale$ as building blocks of digital cash systems}\label{sec:digital-cash-system}

Using $\zeth$ as base application of $\zecale$ as in~\cref{subsec:aggr-market} allows to gain sender anonymity and diminish the amount of data sent and processed on-chain. As mentioned in~\cite[Remark A.2.2]{zeth-specs}, $\zeth$ has been designed to be used in various setting, and the separate addressing scheme of the protocol~\cite[Section 1.4]{zeth-specs} allows to distinguish between users having an $\ethereum$ account and those with only a $\zeth$ address. While $\ethereum$ account holders can trigger arbitrary state transitions, using $\zecale$ along with $\zeth$ is a way to expose the $\zeth$ state transition -- to users holding only a $\zeth$ address -- through relay and aggregator nodes.
We believe that such mechanism is a step toward implementing digital cash systems on blockchains like $\ethereum$, but stress that sound crypto-economics are fundamental to keep the system secure. The possibility to retrieve funds from the $\zeth$ contract in the form of a public output value --- while carrying out a privacy preserving transfer --- certainly is a mechanism that can be used to define and build such incentive structures.

\paragraph{A note on latency.}

An obvious drawback of composing protocols such as $\zeth$, $\zecale$ and $\ethereum$ is the increased settlement latency (i.e.~the time between the emission of the transaction by the user and the processing of the transaction on-chain).
In fact, it is clear that sending zk-SNARK proofs to a relay node over an anonymous network in order to relay them on the blockchain adds latency to execution of the state transition on the distributed state machine (i.e.~blockchain). Any extra computation carried out on the relay node -- like the proof aggregation -- increases this latency even more.

Importantly, high latency in the system may discourage users from using it. As such, and in addition to aggressively optimizing their infrastructure and software to generate wrapping proofs (see~\cref{subsec:aggr-market}), operators of aggregation nodes may be willing to pay high gas price to fasten the inclusion of their aggregation transactions in the blockchain. In a way that echoes and resembles~\cite{coda-economics}, operators of $\zecale$ aggregator nodes would bid to maximize their chances to see their aggregate transactions mined first on the chain. Such high gas price may decrease the aggregation profits but could allow to capture more market shares by offering lower latency access to the $\zeth$ contract to carry out ``fast'' cash payments.

\subsection{Towards better cash}

Today, cash still plays a fundamental role in the economy. Despite a decline in the number of cash payments due to alternative and friction-less digital payment methods, cash remains hugely important, especially for the numerous ``unbanked'' people~\cite{future-cash}.

While transitioning to digital cash seems unavoidable, this transition presents interesting opportunities but also numerous challenges ans risks~\cite{Tanaka_1996,digital-cash}.

A digital cash system as described in~\cref{sec:digital-cash-system} is \emph{secure}, \emph{untraceable}, \emph{distributed}\footnote{the central counter-party processing the payments is replaced by a distributed ledger}, but does not, however, allow to carry out payments at no fees (which is the case for cash). Nevertheless, not only exposing the $\zeth$ state transition via an anonymous aggregation network can grant restricted and secure access to the distributed ledger to people that do not hold an account on the ledger (but only a $\zeth$ address), but it also makes it possible to \emph{implement Anti-Money Laundering (AML) policies for cash payments by modifying the $\zeth$ language in order to (additionally) prove -- in zero-knowledge -- that a payment satisfies a compliance predicate (without disclosing the payment details)}.

We believe that being able to securely grant wide access to the $\zeth$ state transition as well as cryptographically enforce financial regulatory policies is a step toward solving some of the drawbacks of cash and could pave the way to build a more efficient and secure but also more inclusive and stable economy.

\section{Implementation}

In this section we provide an overview of the software architecture of $\zecale$. We invite the reader to consult the open-source project\footnote{\url{https://github.com/clearmatics/zecale}} for more details on the implementation.

\subsection{Ethereum client}

As mentioned in~\cref{subsection-recursion}, any pairing friendly cycle/amicable chain of elliptic curves can be used in $\zecale$ to generate the ``wrapping'' SNARK proof (over one curve) certifying the correct verification of a set of ``nested'' zk-SNARK proofs (generated over another curve). Unfortunately, the current version of $\ethereum$ does not offer the possibility to carry out arithmetic over a wide class of elliptic curves\footnote{Only arithmetic over $\BNCurve$ and $\BLSZcash$ is exposed via precompiled contracts at the time of writing:~\url{https://github.com/ethereum/go-ethereum/blob/0c82928981028e8b32b5852c38b095d2e0d26b04/core/vm/contracts.go\#L83}}.

As of today, implementing $\zecale$ requires to fork from $\ethereum$ mainnet, and add new precompiled contracts to support the $\BWSix$ pairing group operations.
Nevertheless, Ethereum Improvement Proposals such as~\cite{eip1962} may change this situation and expose a wide class of pairing groups via precompiled contracts.

\begin{remark}
Additional EIPs such as~\cite{eip2028} have recently been incorporated to the $\ethereum$ protocol, making layer 2 scalability solutions like $\zecale$ more efficient.
\end{remark}

\subsection{$\zecale$ aggregator}

Following a similar software architecture as in $\zeth$~\cite[Section 6]{DBLP:journals/corr/abs-1904-00905}, we decided to implement the $\zecale$ aggregator as a self-contained software component written in C++ and using a modified version of the $\libsnark$ and $\libff$ libraries\footnote{see~\url{https://github.com/clearmatics/libsnark} and~\url{https://github.com/clearmatics/libff}}. This software component exposes a Remote Procedure Call (RPC) interface allowing to receive zk-SNARK proofs to aggregate. The received zk-SNARK proofs and corresponding instances are then added in an ``application pool'' which represents the set of ``nested'' zk-SNARK proofs to aggregate for a given SNARK-based application.

In addition to receive zk-SNARK proofs to aggregate, other endpoints have been added to the Application Programming Interface (API) allowing to register applications (i.e.~deposit the $\crs$ associated to a given application), along with the back-end logic to enable a $\zecale$ aggregator to function across several applications. A high-level representation of the software components composing $\zecale$ is provided~\cref{fig:zecale-flow}.

\begin{figure}[!htb]
\centering
\includegraphics[width=1\textwidth]{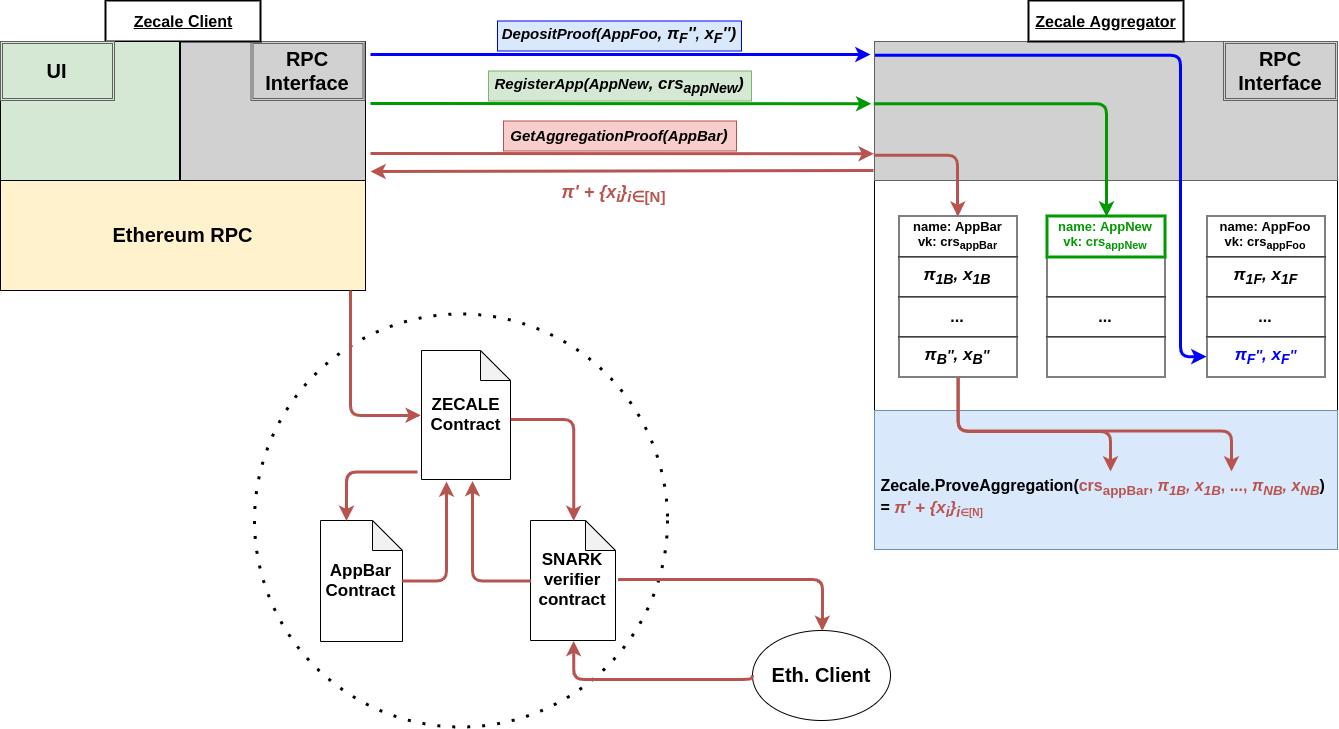}
\caption{An overview of the architecture and flow of function calls for $\zecale$}\label{fig:zecale-flow}
\end{figure}

\subsection{$\zecale$ contract}

A pseudo-code of the smart-contract logic is provided~\cref{fig:zecale-smart-contract}. To be deployed on $\ethereum$, such contract can be implemented using the Solidity programming language\footnote{https://solidity.readthedocs.io/en/v0.7.0/} for instance.

We note that, as previously mentioned, an explicit check --- verifying that all received nested instances lie in the right finite field --- is done by the application contract in $\zecale$. This check is kept on the smart-contract rather than ``moved'' as an extra constraint in $\RELZECALE$ for efficiency reasons. In fact, $\zecaleRelation$ is defined over the scalar field $\FF_\rWrapping$ of $\BWSix$ which is also the base field of $\BLSZexe$. That means that all variables/wires manipulated in the algebraic representation of the $\zecale$ NP-relation are defined over $\FF_\rWrapping$, while the ``nested instances'' are defined over the scalar field $\FF_\rNested$ of $\BLSZexe$. As such, for a given $x$, checking that $x \in \FF_\rNested$ in $\zecaleRelation$ incurs a non-negligible overhead due to the characteristic mismatch of the two scalar field, i.e.$\rWrapping \neq \rNested$. Hence, to bypass the necessity to provably carry out an expensive range-check on the instance, we decided to keep this check on the application smart-contract.

Importantly, since any entry in a set of ``nested instances'' lying outside of the field $\FF_\rNested$ will cause the $\zecale$ state transition to abort, implementers of $\zecale$ aggregator services may --- at their will --- decide to carry out this membership test on ``nested instances'' before accepting the incoming $(\pi, \inp)$ pair from a user. As such, all requests containing ``nested instances'' that are not elements of $\FF_\rNested$ will be deemed invalid and end up being rejected by the aggregation service. Alternatively, the smart-contract logic (as illustrated~\cref{fig:zecale-smart-contract} and~\cref{fig:zecale-smart-contract-dispatch}) can be modified to ignore/skip instances failing to satisfy the field membership test, without aborting.

\medskip

Finally, we note that elements of the base field of $\BLSZexe$ and of the base and scalar field of $\BWSix$ have a binary representation that exceeds the 256-bit word length of the $\ethereum$ (stack-based) Virtual Machine (EVM). As such, representing elements of such fields requires multiple entries on the stack, which makes manipulating such elements relatively expensive gas-wise. As such, implementers may want to allocate some time to optimize the implementation of costly operations of the $\zecale$ smart-contract\footnote{Using inline assembly for instance (see~\url{https://solidity.readthedocs.io/en/v0.7.0/assembly.html?highlight=assembly})}, or may want to implement additional precompiled contracts to carry out costly operations natively on the client.

\section{Conclusion}

In this paper we presented $\zecale$, a general purpose SNARK proof aggregator allowing to scale SNARK-based applications on $\ethereum$. We explained how $\zecale$ works and showed that multiple applications could benefit from it. Likewise, various new types of applications and new market opportunities could emerge from this work, and we believe that $\zecale$ and blockchain-based privacy preserving protocols constitute valuable building blocks for digital cash systems.

Importantly, while $\zecale$ diminishes the data sent and processed on-chain, and thus, makes SNARK-based applications more scalable, we emphasize that this work is not \emph{``the''} solution to on-chain scalability, but is rather what we believe to be a step toward building more scalable systems.

\section{Acknowledgments}

We thank Duncan Tebbs for his constructive remarks throughout this work. Likewise, we thank Michal Zajac for helpful comments on a previous version of this work.

\bibliographystyle{alpha}
\bibliography{references}

\appendix
\section{A note on batch verification}\label{appendix:batch-verif}

We start by defining the notion of batch SNARK verification by adapting the definition of batch verification of signatures from Camenisch, Hohenberger and Pedersen~\cite{DBLP:journals/joc/CamenischHP12}.

\begin{definition}[Batch verification of SNARKs]\label{def:batch-snark}
    Let $\secpar$ be the security parameter and $\REL$ be an NP-relation which associated language is $\LAN$. Suppose $\snark = (\kgen, \prover, \verifier, \simulator)$ is a SNARK scheme, $n \in \poly[\secpar]$, and let $\crs \gets \snark.\kgen(\REL)$. We call probabilistic $\snarkBatch$, a batch verification algorithm where the following conditions hold:
    \begin{description}
        \item[Batch-Completeness:]
        \begin{align*}
            & \snark.\verifier(\crs, \inp_j, \wit_j) = \true\,, \forall j \in [n]\\
            & \Rightarrow \snarkBatch(\crs, (\inp_0, \wit_0), \ldots, (\inp_{n-1}, \wit_{n-1})) = \true
        \end{align*}
        \item[Batch-Soundness:]
        \[
            \condprob{
                \begin{aligned}
                    & \snarkBatch(\crs, (\inp_0, \wit_0), \ldots, (\inp_{n-1}, \wit_{n-1}))\\
                    & = \true
                \end{aligned}
            }{
                \begin{aligned}
                    & \exists\ j \in [n]\ \suchthat\\
                    & \verifier(\crs, \inp_j, \wit_j) = \false
                \end{aligned}
            } \leq \negl[\secpar]
        \]
    \end{description}
\end{definition}

\subsection{Intuition behind the soundness of batching.}\label{subsubsec:batching-framework}

In~\cite{DBLP:conf/eurocrypt/BellareGR98} Bellare et al.~present a set of probabilistic checks to perform fast batch verification for modular exponentiation and digital signatures.
Later, Ferrara et al.~\cite{DBLP:conf/ctrsa/FerraraGHP09} and Blazy et al.~\cite{DBLP:conf/acns/BlazyFIJSV10} proposed batching methods in pairing-based settings.

Similar methods can be used to batch verify a set SNARKs generated for the same language (i.e.~a set of SNARKs generated under the same $\crs$).
As such, instead of checking a set of $N$ proofs separately by running $N$ times the verification algorithm, a single equation can be used to decide whether all SNARKs in the set are valid. If the batch verification equation is not satisfied, this means that \emph{at least} one proof in the set is not valid.

As such it is possible to transform system of verification equations, as below, into a single probabilistic check (we further assume here that the algorithm $\snark.\verifier()$ can be fully represented by a single check that tests that an equation - $\verifierEq$ - equals $0$ if the proof is valid).
\begin{equation}\label{eq:abstract-batch-system}
    \begin{aligned}
        \begin{cases}
            \snark.\verifierEq(\crs, \inp_0, \pi_0) (\isEq 0)\\
            \snark.\verifierEq(\crs, \inp_1, \pi_1) (\isEq 0)\\
            \ldots \\
            \snark.\verifierEq(\crs, \inp_{N-1}, \pi_{N-1}) (\isEq 0)\\
        \end{cases}
    \end{aligned}
\end{equation}

Indeed, representing the system of equations in~\cref{eq:abstract-batch-system} as a vector defined over a vector space, informally allows to ``isolate'' each proof verification into ``its own dimension''. In fact, $\vec{\Pi}$ can be obtained by decomposition in term of the \emph{standard basis} of the vector space, i.e.~:
\begin{equation}
    \vec{\Pi} = f_0 \cdot \begin{bmatrix} 1\\ 0\\ \vdots \\ 0 \end{bmatrix} + \ldots + f_{N-1} \cdot \begin{bmatrix} 0\\ \vdots \\ 0\\ 1 \end{bmatrix}\,
\end{equation}
where $f_i = \snark.\verifierEq(\crs, \inp_i, \pi_i), \in \FF_r$.

Since $\mathcal{B} = \left\{ \begin{bmatrix} 1\\ 0\\ \vdots \\ 0 \end{bmatrix}, \ldots, \begin{bmatrix} 0\\ \vdots \\ 0\\ 1 \end{bmatrix} \right\}$ is a basis of the vector space $\vectorspace{V}$, and by definition of a basis, we know that by the \emph{linear independence property} of the vectors forming a vector space basis, $\vec{\Pi} = \vec{0} \Rightarrow f_i = 0, \forall i \in [N]$, i.e.~all proofs are correct.

While checking if $\vec{\Pi} = \vec{0}$ can be done by checking all coordinates of $\vec{\Pi}$, which corresponds to checking each individual proofs one by one, it is possible to fasten this check by applying a linear transformation $\lintransform{T} \colon (\FF_r)^N \to \FF_r$ defined as: $\vec{M}\vec{x}$, where $\vec{M} = \begin{bmatrix} m_0, \ldots, m_{N-1} \end{bmatrix}$ is a $1 \times N$ matrix defined over $\FF_r$ such that $\exists\ i \in [N]\ \suchthat\ m_i \neq 0$, and $\vec{x} = \begin{bmatrix} x_0, \ldots, x_{N-1} \end{bmatrix}^\top\ \in \vectorspace{V}$.

Applying this linear transformation on $\vec{\Pi}$ allows to ``project the vector to the number line'' and simply check that the transformation input is mapped to 0 by doing a single check.
Trivially, $\vec{x} = \vec{0} = \begin{bmatrix} 0, 0, \ldots, 0 \end{bmatrix}^\top \in (\FF_r)^N$ is mapped to $0$ by $\lintransform{T}$. However, $\Ker(\lintransform{T}) = \{ (x_0, \ldots, x_{N-1}) \in \vectorspace{V}\ |\ m_0 x_0 + \ldots + m_{N-1} x_{N-1} = 0 \mod r\} \neq \{\vec{0}\}$ (i.e.~$\Ker(\lintransform{T}) \supset \{\vec{0}\}$, but $\Ker(\lintransform{T}) \nsubseteq \{\vec{0}\}$). In fact, reducing the dimensions of the space via the transform $\lintransform{T}$ breaks the ``one check $\leftrightarrow$ one dimension'' idea above-mentioned, and it is now important to make sure that no adversary can forge a vector $\vec{\Pi}$ such that $\vec{\Pi} \in \Ker(\lintransform{T}) \setminus \{\vec{0}\}$, i.e.~its projection by the linear transformation above yields the origin of the number line. This would violate the \emph{batch-soundness} property (see~\cref{def:batch-snark}).

To measure the likelihood for an adversary to hold a vector of SNARK proof that violates the soundness of the batch equation, it is necessary to estimate the number of solutions of the equation:
\begin{equation}\label{eq:lin-eq}
    m_0 x_0 + \ldots + m_{N-1} x_{N-1} = 0 \mod r
\end{equation}
, where $m_i$'s and $x_i$'s $\in \FF_r$. As the linear functional $\lintransform{T}$ is assumed to be non-trivial, we know that $\Ker(\lintransform{T}) \subset \vectorspace{V}$. However, knowing how smaller the kernel of the linear functional is compared to the source vector space is crucial to prove the soundness of batching.

Intuitively, given $\lintransform{T}$, the equation~\cref{eq:lin-eq} holds if either all terms are equal to $0$, or if 2 terms - out of the $N$ terms - ``cancel each other'', i.e.~are additive inverse in the field, or if 3 terms ``cancel each other'' etc.
This means that the number of solutions of~\cref{eq:lin-eq} is bounded by the sum of binomial coefficients,
\[
    \sum_{k=0}^{N} \binom{N}{k} = 2^N
\]

As we now know that $\cardinality{\Ker(\lintransform{T})} < 2^N,\ \forall \lintransform{T} \sample \mathfrak{V}$, where $\mathfrak{V}$ is the set of elements in the algebraic dual space of $\vectorspace{V}$ minus the trivial map $\lintransform{T}_t$ (where $\Ker(\lintransform{T}_t) = \vectorspace{V}$), on challenge a random - non-trivial - linear functional, the probability for an adversary $\adv$ (who previously committed to his vector of proofs - to ``kill'' adaptivity) to hold a vector of proofs such that it falls in $\Ker(\lintransform{T})$ is:
\begin{align*}
    \frac{\cardinality{\Ker(\lintransform{T})}}{\cardinality{\vectorspace{V}}} & < \frac{2^N}{(\cardinality{\FF_r})^N} \\
    & < \left(\frac{2}{r}\right)^N
\end{align*}

If $r$ is a large prime such that $r$ is encoded on $\secpar$ bits (i.e. $r > 2^\secpar$), the probability that the vector of proofs held by the prover lies in the kernel of the challenged linear functional is bounded by $\left(\frac{2}{2^\secpar}\right)^N = \frac{1}{2^{N(\secpar -1)}}$, which is negligible in $\secpar$.\footnote{We can also use~\cite[Lemma 1]{MR360598} to bound $\cardinality{\Ker(\lintransform{T})}$ (described as the set of solutions of a multivariate polynomial of degree 1) by $r^{n-1}$, and show that the probability to have a vector that lies in $\Ker(\lintransform{T})$ is bounded by $\frac{r^{n-1}}{r^n} = \frac{1}{r} < \frac{1}{2^\secpar}$.}

\subsubsection{Modeling proof-batching as an interactive protocol.}

Below, we model proofs batching as an interactive protocol between a prover and a verifier.
Informally, the prover has a vector of SNARKs and wants to show the verifier that all proofs are valid.


We follow the same structure as a sigma protocol~\cite{sigma-protocols}, where the prover wants to convince the verifier that he holds a vector of valid SNARKs (and associated public inputs) for a given language - represented by its $\crs$ - which we assume to be available by all parties in the protocol. First, the prover commits to his set of pairs of proofs and associated public inputs, by sending a commitment to the verifier, then the verifier samples a challenge at random from the challenge space, and the prover answers this challenge. At the end of the protocol, the verifier either accepts or rejects.

The protocol is summarized below - where $\Pi_i = (\pi_i, \inp_i)\ \forall i \in [N]$:

\begin{center}
\fbox{
\procedure{Interactive protocol for SNARK proof batching}{
 \textbf{Prover}  \< \crs \gets \snark.\kgen(\REL) \< \textbf{Verifier} \\
 c \gets \algostyle{Commit}(\Pi) \<\< \\
 \< \sendmessageright*{c} \< \\
 \<\< \lintransform{T} \sample \mathfrak{V} \\
 \< \sendmessageleft*{\lintransform{T}} \< \\
 res \gets \lintransform{T}\left(\begin{bmatrix} \snark.\verifierEq(\crs, \Pi_0), \ldots \end{bmatrix}^\top\right) \<\< \\
 \< \sendmessageright*{res} \< \\
 \<\< \algostyle{Decide}(c, \lintransform{T}, res)
}}
\end{center}

\subsubsection{Applying Fiat-Shamir to remove interactions.}

In most cases, the vector $\Pi$ of proofs and primary inputs is held by an aggregator or verifier, and is built as the set of received proof/inputs from a set of provers. As such, the (batch) verifier is \emph{not} one of the initial prover who generated one of the proofs in the batch. We are interested here, in the case where this distinction \emph{does not apply anymore}, and where the batch verifier \emph{is} one of the provers who initially generated a proof in the batch.

We consider the interactive protocol above, and make it non-interactive by using the Fiat-Shamir transform~\cite{DBLP:conf/crypto/FiatS86,DBLP:conf/eurocrypt/PointchevalS96,DBLP:conf/asiacrypt/BernhardPW12}, which yields the following - non-interactive - protocol:
\begin{center}
\fbox{
\procedure{Non-interactive protocol for SNARK proof batching}{
 \textbf{Prover}  \< \crs \gets \snark.\kgen(\REL) \< \textbf{Verifier} \\
 c \gets \algostyle{Commit}(\Pi) \<\< \\
 \lintransform{T} \gets \algostyle{GenChallenge}(c) \\
 res \gets \lintransform{T}\left(\begin{bmatrix} \snark.\verifierEq(\crs, \Pi_0), \ldots \end{bmatrix}^\top\right) \<\< \\
 \< \sendmessageright*{res} \< \\
 \<\< \algostyle{Decide}(c, \lintransform{T}, res)
}}
\end{center}

Here the verifier's challenge is derived by the prover from the first message sent in the interactive protocol. We note that the derivation of $\lintransform{T}$ can be done by calling a Random Oracle (RO) $N$ times to obtain each entry of the line matrix $\vec{M}$, representing $\lintransform{T}$, where each $m_i$ of $\vec{M}$ is obtained as $m_i \gets \algostyle{RO}(c || i)$.
However, doing so, requires to call the RO $N$ times, which may not be desired. As such, one may want to derive $\vec{M}$ as $m_0 \gets \algostyle{RO}(c)$, and derive the other $m_i, i \in [N] \setminus {0}$ as $m_i \gets m_0^{i+1}$. Now, the equation that determines the kernel of the linear transformation becomes a univariate polynomial of degree at most $N$, which, as we know by the (DeMillo-Lipton-)Schwartz-Zippel lemma~\cite{DBLP:journals/ipl/DemilloL78,MR575692,MR594695}, has at most $N$ roots.
In fact, we would have:
\[
    \Ker(\lintransform{T}) = \{ (x_0, \ldots, x_{N-1}) \in \vectorspace{V}\ |\ x_0 m_0 + x_1 m_0^2 + \ldots + x_{N-1} m_0^{N} = 0 \mod r\}
\]

which can be modeled as the set of all roots of the polynomial of degree $N$ which coefficients are $(x_0, \ldots, x_{N-1})$. Schwartz-Zippel tell us that such polynomial has at most $N$ roots, as such, the only way for the equation to be satisfied (and the batch verification check to be accepted) is either,
\begin{enumerate}
    \item to sample a challenge $\algostyle{RO}(c)$ that lends in the set of roots of the polynomial determined by the prover's input vector. The probability of this event is $\frac{N}{\cardinality{\FF_r}}$ which is negligible for low degree polynomials - $N \ll r$ (we note that $N \ll r$ will always be satisfied in applied settings, as carrying out a batch SNARK verification on $N \approx r$ SNARKs would be prohibitively expensive. Alternative protocols such as~\cite{cryptoeprint:2019:1177} would be more efficient to generate the ``wrapping proof'' at the expense of having a log-sized proof with a log-time verifier), or
    \item for the prover's polynomial to be the 0 polynomial (i.e.~for all his proofs to be valid - all checks pass and all entries in the vector defining the polynomial are 0)
\end{enumerate}

Importantly, the malicious prover who tries to violate the soundness of the batch verification, by continuously sampling a new vector of proofs and the corresponding challenge until the batch verification pass while the batch is invalid, does not have a probability higher than $N/\cardinality{\FF_r}$ to succeed since all these events are \emph{independent}. In fact, by re-running the protocol in the hope to violate the \emph{batch-soundness}, the prover generates a new set of proofs and thus a new polynomial to evaluate after receiving the challenge from the random oracle (which output is unpredictable by the prover). Hence, the prover cannot evaluate the same polynomial at multiple points.

All the operations carried out by the prover in the non-interactive protocol above may be added to $\RELZECALE$ in order to be ``carried out in the SNARK''.

\subsection{Batch verification equation for $\groth$.}

We know that the verification routine ran by $\verifier$ in~\cite{DBLP:conf/eurocrypt/Groth16} consists in checking if~\cref{eq:groth-verif} holds.

\begin{equation}\label{eq:groth-verif}
    \pair(\pi.\piA, \pi.\piB) \isEq \pair(\vk.\alpha, \vk.\beta) * \pair(\Gamma, \vk.\gamma) * \pair(\pi.\piC, \vk.\delta)
\end{equation}

where
\[
    \Gamma = \sum_{i \in [\cardinality{\inp}]} \inp_i \frac{\vk.\beta \cdot \polU_i(x) + \vk.\alpha \cdot \polV_i(x) + \polW_i(x)}{\vk.\gamma}
\]

Carrying out this check requires the verifier to compute $4$ pairings along with $\cardinality{\inp}$ scalar multiplications (i.e.~a scalar multiplication for each primary input in $\inp$).

Using batch verification allows to save a few pairing operations for the verification of multiple SNARKs.

Applying the reasoning described in the first section of~\cref{subsubsec:batching-framework}, we can represent a set of $\groth$ proof verification equations

\begin{equation}\label{eq:batch-system}
    \begin{aligned}
        \begin{cases}
            \piA_1 \piB_1 = \alpha \beta + \gamma \Gamma_1 + \delta \piC_1\\
            \piA_2 \piB_2 = \alpha \beta + \gamma \Gamma_2 + \delta \piC_2\\
            \ldots \\
            \piA_N \piB_N = \alpha \beta + \gamma \Gamma_N + \delta \piC_N\\
        \end{cases}&\Leftrightarrow&
        \begin{cases}
            \piA_1 \piB_1 - \alpha \beta - \gamma \Gamma_1 - \delta \piC_1 = 0\\
            \piA_2 \piB_2 - \alpha \beta - \gamma \Gamma_1 - \delta \piC_2 = 0\\
            \ldots \\
            \piA_N \piB_N - \alpha \beta - \gamma \Gamma_N - \delta \piC_N = 0\\
        \end{cases}
    \end{aligned}
\end{equation}

as a vector $\vec{\Pi}$ defined over the vector space $\vectorspace{V} = (\FF_r)^N$, where $r$ is the prime order of $\GRP_1, \GRP_2, \text{and}\ \GRP_T$:
\begin{align}
    \vec{\Pi} &= \begin{bmatrix}
            \piA_1 \piB_1 - \alpha \beta - \gamma \Gamma_1 - \delta \piC_1\\
            \piA_2 \piB_2 - \alpha \beta - \gamma \Gamma_1 - \delta \piC_2\\
            \vdots \\
            \piA_N \piB_N - \alpha \beta - \gamma \Gamma_N - \delta \piC_N\\
         \end{bmatrix}
\end{align}

and batch verify the set of underlying proofs.
The goal of the verifier is to check that $\vec{\Pi}$ is equal to $\vec{0}$.


\paragraph{Groth16 batched - unrolled.}

We can transform this set of equations into a single equation, by applying a random linear functional $\lintransform{T}$, represented by a line matrix of random scalars $\vec{M} = \begin{bmatrix} m_0, \ldots, m_{N-1} \end{bmatrix}$, and checking that the result is $0 \in \FF_r$, i.e.~checking that $\lintransform{T}(\Pi) = \vec{M}\Pi = 0 \in \FF_r$.
After applying the group encoding on field elements and denoting by $\vk$ (``verification key'') the part of the $\crs$ used by the verifier, the batch verification check becomes:

\begin{equation}
    \prod_{i \in [N]} \pair(\pi_i.\piA, \pi_i.\piB)^{m_i} \isEq \prod_{i \in [N]} \left[\pair(\vk.\alpha, \vk.\beta) * \pair(\Gamma_i, \vk.\gamma) * \pair(\pi_i.\piC, \vk.\delta)\right]^{m_i}
\end{equation}

Which gives:

\begin{equation}\label{eq:groth-batch-verify}
    \prod_{i \in [N]} \pair\left(m_i \cdot \pi_i.\piA, \pi_i.\piB\right) \isEq \pair\left(\textstyle\sum_{i \in [N]} m_i \cdot \vk.\alpha, \vk.\beta\right) * \pair\left(\widetilde{\Gamma}, \vk.\gamma\right) * \pair\left(\textstyle\sum_{i \in [N]} m_i \cdot \pi_i.\piC, \vk.\delta\right)
\end{equation}

where:

\[
    \widetilde{\Gamma} = \sum_{j \in [\cardinality{\inp}]} \frac{\left[\vk.\beta \cdot \polU_j(x) + \vk.\alpha \cdot \polV_j(x) + \polW_j(x)\right] (\sum_{i \in [N]} m_i \cdot \inp_j)}{\vk.\gamma}
\]

In contrast with~\cref{eq:groth-verif}, the batch verification equation~\cref{eq:groth-batch-verify} requires $N + 3$ pairings (instead of $4 \times N$) along with $1 + \cardinality{\inp} + 2N$ scalar multiplications (instead of $N \times \cardinality{\inp}$) to verify $N$ SNARKs.
As such, a few pairing computations are removed at the expense of more scalar multiplications, which as we know, are significantly cheaper to carry out.

\section{Practical considerations of SNARK batching}\label{appendix:practical-considerations}

While it is possible to batch verify a set of nested proofs in \zecale, one needs to understand the limitations of such technique. In fact, batch verification presents some challenges related forgeries detection in batches, which can ultimately, significantly slow the system down.

\subsection{Forgery detection.}

As above-mentioned, the batch verification equation for $\groth$ can be used to efficiently check that all entries in a set of SNARKs are valid. In other words, the batch verification method \emph{detects the presence of forgeries in a set} of SNARKs.
Nevertheless, while a failing batch verification exposes the presence of forgeries in the set of inputs, it does not allow to identify which input - in the set - is a forgery. While a trivial way to identify forgeries consists in verifying all inputs individually, being able to identify forgeries in a batch efficiently is non-trivial and has lead to various work.

\medskip

While ``divide-and-conquer'' approaches such as the one proposed by Pastuszak et al.~\cite{DBLP:conf/pkc/PastuszakMS00} can be used to converge toward forgeries in the set of inputs, we note that this method introduces an important overhead on the system. Later, Law and Matt~\cite{DBLP:conf/ima/LawM07,DBLP:conf/pkc/Matt09} proposed an improved forgery identification protocol for pairing-based signatures, improving on the initial design of binary-tree search methods.

While studying SNARK forgery identification is out of scope of this work, we note that carrying out forgery identification upon invalid SNARK batch verification gives a leverage to an adversary $\adv$ who can send a low volume of SNARK forgeries to the aggregator, enough to have one forgery in each batch, in order to cause a severe slow down in the system~\cite{DBLP:conf/indocrypt/BernsteinDLO12}.

As such, some care needs to be taken before considering using batch SNARK verification as an optimization in protocols like $\zecale$.

\subsection{Processing and verifying invalid proofs.}

For systems like $\zecale$, where an aggregator piece of software receives a set of proofs and relays the result of their aggregation on-chain, it may be desired to process invalid proofs. In fact, settling a transaction containing an invalid proof on a blockchain system acts as a proof that the transaction has been processed - and not censored. As such, even if a SNARK does not verify successfully on an aggregator, one may still want to send the result of this erroneous verification on-chain to show the progress of the system, and to inform the transaction originator that the transaction has been processed and added on-chain.


\end{document}

%% file: config/packages.tex
\usepackage[utf8]{inputenc}
\usepackage[english]{babel}
\usepackage{graphicx}
\usepackage{xcolor}
\usepackage{float}
\usepackage{amsmath}
\usepackage{amssymb}
\usepackage[normalem]{ulem}
\usepackage[color=yellow!20]{todonotes}
\usepackage{soul}
\usepackage{hyperref}
\usepackage{paralist}
\usepackage{mathrsfs}
\usepackage{listings}
\usepackage{lscape} 
\usepackage{dashbox} 

\usepackage{enumitem} 
\newcommand{\subscript}[2]{$#1 _ #2$} 

\usepackage[paperheight=23.5cm,paperwidth=15.5cm,text={12.2cm,19.3cm},centering]{geometry}
\usepackage[innerleftmargin=5pt,innerrightmargin=5pt]{mdframed}
\usepackage{bm}
\usepackage[capitalise]{cleveref}
\usepackage{caption}
\usepackage{subcaption}
\usepackage[rightcaption]{sidecap}
\usepackage[mathscr]{euscript}
\usepackage[all]{nowidow}
\usepackage{comment}

\usepackage[advantage, asymptotics, adversary, complexity, sets, keys, ff, notions, lambda, primitives, events, operators, probability, logic, mm, landau]{cryptocode}

%% file: config/macros.tex
\newcommand{\variable}[1]{\mathsf{#1}}
\newcommand{\constant}[1]{\mathtt{#1}}

\renewcommand{\pccomplexitystyle}[1]{\ensuremath{\mathsf{#1}}}
\newcommand{\IP}{\pccomplexitystyle{IP}}
\newcommand{\PSPACE}{\pccomplexitystyle{PSPACE}}

\newcommand{\project}[1]{\texttt{#1}}
\newcommand{\zeth}{\project{ZETH}}
\newcommand{\zecale}{\project{Zecale}}
\newcommand{\zexe}{\project{ZEXE}}
\newcommand{\ethereum}{\project{Ethereum}}
\newcommand{\coda}{\project{Coda}}
\newcommand{\nym}{\project{Nym}}
\newcommand{\tor}{\project{TOR}}
\newcommand{\loopix}{\project{Loopix}}
\newcommand{\dizk}{\project{DIZK}}
\newcommand{\libsnark}{\project{libsnark}}
\newcommand{\libff}{\project{libff}}

\newcommand{\zecaleP}{\algostyle{ZecaleP}} 
\newcommand{\appP}{\algostyle{BaseAppP}} 

\newcommand{\asset}[1]{\mathtt{#1}}
\newcommand{\ether}{\asset{Ether}}

\let\vec\bm%

\newcommand{\accountstyle}[1]{\mathbf{#1}}
\newcommand{\contractstyle}[1]{\widetilde{\accountstyle{#1}}}
\newcommand{\baseAppContract}{\ensuremath{\contractstyle{baseAppC}}}
\newcommand{\modifiedBaseAppContract}{\ensuremath{\contractstyle{ZbaseAppC}}}
\newcommand{\ZecaleContract}{\ensuremath{\contractstyle{ZecaleC}}}

\newcommand{\ledger}{\mathscr{L}}

\newcommand{\zktx}{\variable{zktx}} 
\newcommand{\aggrtx}{\variable{aggrtx}} 
\newcommand{\inputs}{\variable{inputs}} 

\newcommand{\execAppLogic}{\algostyle{ExecAppLogic}} 
\newcommand{\zkAppCRS}{\variable{\crs_{app}}} 
\newcommand{\zkOtherAppCRS}{\variable{\crs_{\widetilde{app}}}}
\newcommand{\zecaleCRS}{\variable{\crs_{zec}}}
\newcommand{\zecaleRelation}{\algostyle{ZecaleRelation}}
\newcommand{\zkpZecale}{\variable{\pi_{zec}}}
\newcommand{\zkpBaseApp}{\variable{\pi_{app}}}

\newcommand{\inpZecale}{\variable{\inp_{zec}}}
\newcommand{\inpBaseApp}{\variable{\inp_{app}}}

\newcommand{\snarkZecale}{\algostyle{\snark_{zec}}}
\newcommand{\snarkApp}{\algostyle{\snark_{app}}}

\newcommand{\appContractAddress}{\variable{baseAppAddr}}
\newcommand{\ZappContractAddress}{\variable{ZbaseAppAddr}}
\newcommand{\zecaleContractAddress}{\variable{zecaleAddr}}
\newcommand{\dispatchData}{\variable{dispatchData}}
\newcommand{\dispatch}{\algostyle{dispatch}}
\newcommand{\createContractInstance}{\algostyle{createContractInstance}}
\newcommand{\gasSaved}{\variable{gSaved}}
\newcommand{\snarkBatch}{\algostyle{VBATCH}} 
\newcommand{\verifierEq}{\algostyle{\verifier_{eq}}}

\newcommand{\inpPacked}{\variable{xH}}
\newcommand{\inpValidity}{\variable{xValid}}
\newcommand{\processAggregatedTx}{\algostyle{processAggrTx}}
\newcommand{\processTx}{\algostyle{processTx}}
\newcommand{\constructor}{\algostyle{constructor}}
\newcommand{\storage}{\variable{storage}}
\newcommand{\vkHash}{\variable{vkHash}}

\newcommand{\encodeToBytes}{\algostyle{encodeToBytes}}
\newcommand{\byteData}{\variable{byteData}}
\newcommand{\msgSender}{\variable{msgSender}}
\newcommand{\decodeBytes}{\algostyle{decodeBytes}}
\newcommand{\decodedData}{\variable{decodedData}}
\newcommand{\fieldAppCRS}{\variable{``app\_crs"}}
\newcommand{\fieldZecaleAddr}{\variable{``zecale\_addr"}}
\newcommand{\toDigest}{\algostyle{toDigest}}
\newcommand{\toField}{\algostyle{toField}}
\newcommand{\rNested}{\variable{r_n}}
\newcommand{\rWrapping}{\variable{r_w}}

\newcommand{\batchSize}{\constant{BATCH\_SIZE}} 
\newcommand{\verifNProofGas}{\constant{VNProofGas}} 
\newcommand{\verifWProofGas}{\constant{VWProofGas}} 
\newcommand{\txDefaultGas}{\constant{DGAS}}

\newcommand{\game}[1]{\mathsf{#1}}
\newcommand{\zclSND}{\game{ZCL\mhyphen SND}} 
\newcommand{\collRes}{\game{coll\mhyphen res}} 
\newcommand{\snarkSND}{\game{SNARK\mhyphen SND}} 

\newcommand{\curve}[1]{\mathsf{#1}}
\newcommand{\BNCurve}{\curve{BN\mhyphen{}254}}
\newcommand{\BLSZcash}{\curve{BLS12\mhyphen{}381}}
\newcommand{\BLSZexe}{\curve{BLS12\mhyphen{}377}}
\newcommand{\BWSix}{\curve{BW6\mhyphen{}761}}
\newcommand{\CPcurve}{\curve{CP}}

\newcommand{\algostyle}[1]{\mathsf{#1}}
\newcommand{\oracle}[1]{\algostyle{O}^{#1}} 

\newcommand{\setup}{\algostyle{Setup}}
\newcommand{\includeTxInBlock}{\algostyle{IncludeTxInBlock}}

\DeclareMathOperator\shr{shr}

\newcommand{\snark}{\algostyle{\Psi}}
\newcommand{\LAN}{\mathbf{L}}
\newcommand{\LANZECALE}{\LAN^{\project{zec}}}
\newcommand{\REL}{\mathbf{R}}
\newcommand{\RELZECALE}{\REL^{\project{zec}}}

\newcommand{\transcript}{\mathsf{Transcript}}
\newcommand{\groth}{\algostyle{Groth16}}
\newcommand{\piA}{\variable{A}} 
\newcommand{\piB}{\variable{B}} 
\newcommand{\piC}{\variable{C}} 
\newcommand{\polU}{\variable{u}} 
\newcommand{\polV}{\variable{v}} 
\newcommand{\polW}{\variable{w}} 
\newcommand{\crs}{\variable{crs}} 

\newcommand{\td}{\variable{td}} 
\newcommand{\inp}{\variable{x}} 
\newcommand{\wit}{\variable{w}} 

\newcommand{\paramGen}{\mathcal{G}}
\newcommand{\suchthat}{\text{s.t.}} 
\newcommand{\GRP}{\mathbb{G}} 
\newcommand{\GRPord}{r} 
\newcommand{\pair}{e} 
\newcommand{\ggen}{\mathfrak{g}} 
\newcommand{\isEq}{\stackrel{?}{=}}
\newcommand{\cardinality}[1]{|#1|}
\newcommand{\vectorspace}[1]{\mathcal{#1}} 
\newcommand{\lintransform}[1]{\mathcal{#1}} 
\DeclareMathOperator{\Ker}{Ker}

\mathchardef\mhyphen="2D

\definecolor{bananamania}{rgb}{0.98,0.91,0.71}
\definecolor{darkred}{rgb}{0.7,0,0}

%% file: zecale.bbl
\newcommand{\etalchar}[1]{$^{#1}$}
\begin{thebibliography}{BMMV19}

\bibitem[Adi06]{ba-phd-thesis}
Ben Adida.
\newblock Phd thesis.
\newblock \url{http://assets.adida.net/research/phd-thesis.pdf}, 2006.
\newblock Accessed: 2020-05-27.

\bibitem[ASB{\etalchar{+}}19]{eip2028}
Alexey Akhunov, Eli~Ben Sasson, Tom Brand, Louis Guthmann, and Avihu Levy.
\newblock Eip 2028: Transaction data gas cost reduction.
\newblock \url{https://eips.ethereum.org/EIPS/eip-2028}, 2019.

\bibitem[Aut16]{eth-sharding}
Ethereum~Wiki Authors.
\newblock Sharding faq.
\newblock \url{https://github.com/ethereum/wiki/wiki/Sharding-FAQ}, 2016.
\newblock Accessed: 2020-05-27.

\bibitem[Bab85]{DBLP:conf/stoc/Babai85}
L{\'{a}}szl{\'{o}} Babai.
\newblock Trading group theory for randomness.
\newblock In Robert Sedgewick, editor, {\em Proceedings of the 17th Annual
  {ACM} Symposium on Theory of Computing, May 6-8, 1985, Providence, Rhode
  Island, {USA}}, pages 421--429. {ACM}, 1985.

\bibitem[Bar18]{barry-rollup}
BarryWhiteHat.
\newblock rool\_up.
\newblock \url{https://github.com/barryWhiteHat/roll_up}, 2018.

\bibitem[BCCT13]{DBLP:conf/stoc/BitanskyCCT13}
Nir Bitansky, Ran Canetti, Alessandro Chiesa, and Eran Tromer.
\newblock Recursive composition and bootstrapping for {SNARKS} and
  proof-carrying data.
\newblock In Dan Boneh, Tim Roughgarden, and Joan Feigenbaum, editors, {\em
  Symposium on Theory of Computing Conference, STOC'13, Palo Alto, CA, USA,
  June 1-4, 2013}, pages 111--120. {ACM}, 2013.

\bibitem[BCG{\etalchar{+}}14]{DBLP:conf/sp/Ben-SassonCG0MTV14}
Eli Ben{-}Sasson, Alessandro Chiesa, Christina Garman, Matthew Green, Ian
  Miers, Eran Tromer, and Madars Virza.
\newblock Zerocash: Decentralized anonymous payments from bitcoin.
\newblock In {\em 2014 {IEEE} Symposium on Security and Privacy, {SP} 2014,
  Berkeley, CA, USA, May 18-21, 2014}, pages 459--474. {IEEE} Computer Society,
  2014.

\bibitem[BCG{\etalchar{+}}18]{DBLP:journals/iacr/BoweCGMMW18}
Sean Bowe, Alessandro Chiesa, Matthew Green, Ian Miers, Pratyush Mishra, and
  Howard Wu.
\newblock Zexe: Enabling decentralized private computation.
\newblock {\em {IACR} Cryptol. ePrint Arch.}, 2018:962, 2018.

\bibitem[BCI{\etalchar{+}}13]{DBLP:conf/tcc/BitanskyCIPO13}
Nir Bitansky, Alessandro Chiesa, Yuval Ishai, Rafail Ostrovsky, and Omer
  Paneth.
\newblock Succinct non-interactive arguments via linear interactive proofs.
\newblock In Amit Sahai, editor, {\em Theory of Cryptography - 10th Theory of
  Cryptography Conference, {TCC} 2013, Tokyo, Japan, March 3-6, 2013.
  Proceedings}, volume 7785 of {\em Lecture Notes in Computer Science}, pages
  315--333. Springer, 2013.

\bibitem[BCTV17]{DBLP:journals/algorithmica/Ben-SassonCTV17}
Eli Ben{-}Sasson, Alessandro Chiesa, Eran Tromer, and Madars Virza.
\newblock Scalable zero knowledge via cycles of elliptic curves.
\newblock {\em Algorithmica}, 79(4):1102--1160, 2017.

\bibitem[BDLO12]{DBLP:conf/indocrypt/BernsteinDLO12}
Daniel~J. Bernstein, Jeroen Doumen, Tanja Lange, and Jan{-}Jaap Oosterwijk.
\newblock Faster batch forgery identification.
\newblock In Steven~D. Galbraith and Mridul Nandi, editors, {\em Progress in
  Cryptology - {INDOCRYPT} 2012, 13th International Conference on Cryptology in
  India, Kolkata, India, December 9-12, 2012. Proceedings}, volume 7668 of {\em
  Lecture Notes in Computer Science}, pages 454--473. Springer, 2012.

\bibitem[BES16]{falcon-btc}
Soumya Basu, Ittay Eyal, and Emin~Gün Sirer.
\newblock Falcon network.
\newblock
  \url{https://www.falcon-net.org/papers/falcon-retreat-2016-05-17.pdf}, 2016.
\newblock Accessed: 2020-05-27.

\bibitem[BFI{\etalchar{+}}10]{DBLP:conf/acns/BlazyFIJSV10}
Olivier Blazy, Georg Fuchsbauer, Malika Izabach{\`{e}}ne, Amandine Jambert,
  Herv{\'{e}} Sibert, and Damien Vergnaud.
\newblock Batch groth-sahai.
\newblock In Jianying Zhou and Moti Yung, editors, {\em Applied Cryptography
  and Network Security, 8th International Conference, {ACNS} 2010, Beijing,
  China, June 22-25, 2010. Proceedings}, volume 6123 of {\em Lecture Notes in
  Computer Science}, pages 218--235, 2010.

\bibitem[BFM88]{DBLP:conf/stoc/BlumFM88}
Manuel Blum, Paul Feldman, and Silvio Micali.
\newblock Non-interactive zero-knowledge and its applications (extended
  abstract).
\newblock In Janos Simon, editor, {\em Proceedings of the 20th Annual {ACM}
  Symposium on Theory of Computing, May 2-4, 1988, Chicago, Illinois, {USA}},
  pages 103--112. {ACM}, 1988.

\bibitem[BG18]{DBLP:journals/iacr/BoweG18}
Sean Bowe and Ariel Gabizon.
\newblock Making groth's zk-snark simulation extractable in the random oracle
  model.
\newblock {\em {IACR} Cryptology ePrint Archive}, 2018:187, 2018.

\bibitem[BGR98]{DBLP:conf/eurocrypt/BellareGR98}
Mihir Bellare, Juan~A. Garay, and Tal Rabin.
\newblock Fast batch verification for modular exponentiation and digital
  signatures.
\newblock In Kaisa Nyberg, editor, {\em Advances in Cryptology - {EUROCRYPT}
  '98, International Conference on the Theory and Application of Cryptographic
  Techniques, Espoo, Finland, May 31 - June 4, 1998, Proceeding}, volume 1403
  of {\em Lecture Notes in Computer Science}, pages 236--250. Springer, 1998.

\bibitem[BKM{\etalchar{+}}16]{DBLP:journals/jpc/BackesKMMM16}
Michael Backes, Aniket Kate, Praveen Manoharan, Sebastian Meiser, and Esfandiar
  Mohammadi.
\newblock Anoa: {A} framework for analyzing anonymous communication protocols.
\newblock {\em J. Priv. Confidentiality}, 7(2), 2016.

\bibitem[BM88]{Babai1988ArthurMerlinGA}
L.~Babai and S.~Moran.
\newblock Arthur-merlin games: A randomized proof system, and a hierarchy of
  complexity classes.
\newblock {\em J. Comput. Syst. Sci.}, 36:254--276, 1988.

\bibitem[BMMV19]{cryptoeprint:2019:1177}
Benedikt Bünz, Mary Maller, Pratyush Mishra, and Noah Vesely.
\newblock Proofs for inner pairing products and applications.
\newblock Cryptology ePrint Archive, Report 2019/1177, 2019.
\newblock \url{https://eprint.iacr.org/2019/1177}.

\bibitem[BMRS20]{cryptoeprint:2020:352}
Joseph Bonneau, Izaak Meckler, Vanishree Rao, and Evan Shapiro.
\newblock Coda: Decentralized cryptocurrency at scale.
\newblock Cryptology ePrint Archive, Report 2020/352, 2020.
\newblock \url{https://eprint.iacr.org/2020/352}.

\bibitem[BPW12]{DBLP:conf/asiacrypt/BernhardPW12}
David Bernhard, Olivier Pereira, and Bogdan Warinschi.
\newblock How not to prove yourself: Pitfalls of the fiat-shamir heuristic and
  applications to helios.
\newblock In Xiaoyun Wang and Kazue Sako, editors, {\em Advances in Cryptology
  - {ASIACRYPT} 2012 - 18th International Conference on the Theory and
  Application of Cryptology and Information Security, Beijing, China, December
  2-6, 2012. Proceedings}, volume 7658 of {\em Lecture Notes in Computer
  Science}, pages 626--643. Springer, 2012.

\bibitem[BR06]{DBLP:conf/eurocrypt/BellareR06}
Mihir Bellare and Phillip Rogaway.
\newblock The security of triple encryption and a framework for code-based
  game-playing proofs.
\newblock In Serge Vaudenay, editor, {\em Advances in Cryptology - {EUROCRYPT}
  2006, 25th Annual International Conference on the Theory and Applications of
  Cryptographic Techniques, St. Petersburg, Russia, May 28 - June 1, 2006,
  Proceedings}, volume 4004 of {\em Lecture Notes in Computer Science}, pages
  409--426. Springer, 2006.

\bibitem[BSMP91]{DBLP:journals/siamcomp/BlumSMP91}
Manuel Blum, Alfredo~De Santis, Silvio Micali, and Giuseppe Persiano.
\newblock Noninteractive zero-knowledge.
\newblock {\em {SIAM} J. Comput.}, 20(6):1084--1118, 1991.

\bibitem[But14]{ethereum-whitepaper}
Vitalik Buterin.
\newblock Ethereum: A next-generation smart contract and decentralized
  application platform.
\newblock \url{https://github.com/ethereum/wiki/wiki/White-Paper}, 2014.
\newblock Accessed: 2019-08-22.

\bibitem[BW05]{DBLP:journals/dcc/BrezingW05}
Friederike Brezing and Annegret Weng.
\newblock Elliptic curves suitable for pairing based cryptography.
\newblock {\em Des. Codes Cryptogr.}, 37(1):133--141, 2005.

\bibitem[CCW19]{DBLP:journals/siaga/ChiesaCW19}
Alessandro Chiesa, Lynn Chua, and Matthew Weidner.
\newblock On cycles of pairing-friendly elliptic curves.
\newblock {\em {SIAM} J. Appl. Algebra Geom.}, 3(2):175--192, 2019.

\bibitem[Cha03]{DBLP:series/ais/Chaum03}
David Chaum.
\newblock Untraceable electronic mail, return addresses and digital pseudonyms.
\newblock In Dimitris Gritzalis, editor, {\em Secure Electronic Voting},
  volume~7 of {\em Advances in Information Security}, pages 211--219. Springer,
  2003.

\bibitem[CHM{\etalchar{+}}20]{DBLP:conf/eurocrypt/ChiesaHMMVW20}
Alessandro Chiesa, Yuncong Hu, Mary Maller, Pratyush Mishra, Noah Vesely, and
  Nicholas~P. Ward.
\newblock Marlin: Preprocessing zksnarks with universal and updatable {SRS}.
\newblock In Anne Canteaut and Yuval Ishai, editors, {\em Advances in
  Cryptology - {EUROCRYPT} 2020 - 39th Annual International Conference on the
  Theory and Applications of Cryptographic Techniques, Zagreb, Croatia, May
  10-14, 2020, Proceedings, Part {I}}, volume 12105 of {\em Lecture Notes in
  Computer Science}, pages 738--768. Springer, 2020.

\bibitem[CHP12]{DBLP:journals/joc/CamenischHP12}
Jan Camenisch, Susan Hohenberger, and Michael~{\O}stergaard Pedersen.
\newblock Batch verification of short signatures.
\newblock {\em J. Cryptology}, 25(4):723--747, 2012.

\bibitem[Cla18]{future-cash}
David Clarke.
\newblock The future of cash, protecting access to payments in the digital age.
\newblock
  \url{https://positivemoney.org/wp-content/uploads/2018/03/Positive-Money-Future-of-Cash.pdf},
  2018.

\bibitem[Cor]{fibre-btc}
Matt Corallo.
\newblock The fast internet bitcoin relay engine.
\newblock \url{http://www.bitcoinfibre.org/}.
\newblock Accessed: 2020-05-27.

\bibitem[Cor16]{bip-152}
Matt Corallo.
\newblock Bip: 152 - compact block relay.
\newblock \url{https://github.com/bitcoin/bips/blob/master/bip-0152.mediawiki},
  2016.

\bibitem[CP01]{CocksPinch}
C.~Cocks and R.~Pinch.
\newblock Identity-based cryptosystems based on the weil pairing.
\newblock unpublished manuscript, 2001.

\bibitem[CR20]{zeth-specs}
Clearmatics Cryptography~R\&D.
\newblock Zeth protocol specifications.
\newblock \url{https://github.com/clearmatics/zeth-specifications}, 2020.
\newblock Accessed: 2020-07-02.

\bibitem[CRT98]{DBLP:journals/eatcs/ClementiRT98}
Andrea E.~F. Clementi, Jos{\'{e}} D.~P. Rolim, and Luca Trevisan.
\newblock Recent advances towards proving {P} = {BPP}.
\newblock {\em Bulletin of the {EATCS}}, 64, 1998.

\bibitem[CST20]{coda-economics}
Brad Cohn, Evan Shapiro, and Emre Tekişalp.
\newblock Coda: Economics and monetary policy.
\newblock \url{https://codaprotocol.com/static/pdf/economicsWP.pdf}, 2020.

\bibitem[CT10]{DBLP:conf/innovations/ChiesaT10}
Alessandro Chiesa and Eran Tromer.
\newblock Proof-carrying data and hearsay arguments from signature cards.
\newblock In Andrew~Chi{-}Chih Yao, editor, {\em Innovations in Computer
  Science - {ICS} 2010, Tsinghua University, Beijing, China, January 5-7, 2010.
  Proceedings}, pages 310--331. Tsinghua University Press, 2010.

\bibitem[Dam10]{sigma-protocols}
Ivan Damgard.
\newblock On {$\Sigma$} protocols.
\newblock \url{https://www.cs.au.dk/~ivan/Sigma.pdf}, 2010.

\bibitem[DG09]{DBLP:conf/sp/DanezisG09}
George Danezis and Ian Goldberg.
\newblock Sphinx: {A} compact and provably secure mix format.
\newblock In {\em 30th {IEEE} Symposium on Security and Privacy (S{\&}P 2009),
  17-20 May 2009, Oakland, California, {USA}}, pages 269--282. {IEEE} Computer
  Society, 2009.

\bibitem[DH16]{digital-cash}
Ben Dyson and Graham Hodgson.
\newblock Digital cash, why central banks should start issuing electronic
  money.
\newblock
  \url{https://positivemoney.org/wp-content/uploads/2016/01/Digital_Cash_WebPrintReady_20160113.pdf},
  2016.

\bibitem[DL78]{DBLP:journals/ipl/DemilloL78}
Richard~A. DeMillo and Richard~J. Lipton.
\newblock A probabilistic remark on algebraic program testing.
\newblock {\em Inf. Process. Lett.}, 7(4):193--195, 1978.

\bibitem[DMS10]{tor-design}
Roger Dingledine, Nick Mathewson, and Paul Syverson.
\newblock
  \url{https://svn-archive.torproject.org/svn/projects/design-paper/tor-design.pdf},
  2010.
\newblock Accessed: 2020-05-27.

\bibitem[Eth]{zk-rollups}
EthHub.
\newblock Zk-rollups.
\newblock
  \url{https://docs.ethhub.io/ethereum-roadmap/layer-2-scaling/zk-rollups/}.
\newblock Accessed: 2020-05-27.

\bibitem[FGHP09]{DBLP:conf/ctrsa/FerraraGHP09}
Anna~Lisa Ferrara, Matthew Green, Susan Hohenberger, and Michael~{\O}stergaard
  Pedersen.
\newblock Practical short signature batch verification.
\newblock In Marc Fischlin, editor, {\em Topics in Cryptology - {CT-RSA} 2009,
  The Cryptographers' Track at the {RSA} Conference 2009, San Francisco, CA,
  USA, April 20-24, 2009. Proceedings}, volume 5473 of {\em Lecture Notes in
  Computer Science}, pages 309--324. Springer, 2009.

\bibitem[FS86]{DBLP:conf/crypto/FiatS86}
Amos Fiat and Adi Shamir.
\newblock How to prove yourself: Practical solutions to identification and
  signature problems.
\newblock In Andrew~M. Odlyzko, editor, {\em Advances in Cryptology - {CRYPTO}
  '86, Santa Barbara, California, USA, 1986, Proceedings}, volume 263 of {\em
  Lecture Notes in Computer Science}, pages 186--194. Springer, 1986.

\bibitem[FST10]{DBLP:journals/joc/FreemanST10}
David Freeman, Michael Scott, and Edlyn Teske.
\newblock A taxonomy of pairing-friendly elliptic curves.
\newblock {\em J. Cryptology}, 23(2):224--280, 2010.

\bibitem[GGPR13]{DBLP:conf/eurocrypt/GennaroGP013}
Rosario Gennaro, Craig Gentry, Bryan Parno, and Mariana Raykova.
\newblock Quadratic span programs and succinct nizks without pcps.
\newblock In Thomas Johansson and Phong~Q. Nguyen, editors, {\em Advances in
  Cryptology - {EUROCRYPT} 2013, 32nd Annual International Conference on the
  Theory and Applications of Cryptographic Techniques, Athens, Greece, May
  26-30, 2013. Proceedings}, volume 7881 of {\em Lecture Notes in Computer
  Science}, pages 626--645. Springer, 2013.

\bibitem[GKM{\etalchar{+}}18]{DBLP:conf/crypto/GrothKMMM18}
Jens Groth, Markulf Kohlweiss, Mary Maller, Sarah Meiklejohn, and Ian Miers.
\newblock Updatable and universal common reference strings with applications to
  zk-snarks.
\newblock In Hovav Shacham and Alexandra Boldyreva, editors, {\em Advances in
  Cryptology - {CRYPTO} 2018 - 38th Annual International Cryptology Conference,
  Santa Barbara, CA, USA, August 19-23, 2018, Proceedings, Part {III}}, volume
  10993 of {\em Lecture Notes in Computer Science}, pages 698--728. Springer,
  2018.

\bibitem[GM17]{DBLP:conf/crypto/GrothM17}
Jens Groth and Mary Maller.
\newblock Snarky signatures: Minimal signatures of knowledge from
  simulation-extractable snarks.
\newblock In Jonathan Katz and Hovav Shacham, editors, {\em Advances in
  Cryptology - {CRYPTO} 2017 - 37th Annual International Cryptology Conference,
  Santa Barbara, CA, USA, August 20-24, 2017, Proceedings, Part {II}}, volume
  10402 of {\em Lecture Notes in Computer Science}, pages 581--612. Springer,
  2017.

\bibitem[GMR85]{DBLP:conf/stoc/GoldwasserMR85}
Shafi Goldwasser, Silvio Micali, and Charles Rackoff.
\newblock The knowledge complexity of interactive proof-systems (extended
  abstract).
\newblock In Robert Sedgewick, editor, {\em Proceedings of the 17th Annual
  {ACM} Symposium on Theory of Computing, May 6-8, 1985, Providence, Rhode
  Island, {USA}}, pages 291--304. {ACM}, 1985.

\bibitem[GO94]{DBLP:journals/joc/GoldreichO94}
Oded Goldreich and Yair Oren.
\newblock Definitions and properties of zero-knowledge proof systems.
\newblock {\em J. Cryptology}, 7(1):1--32, 1994.

\bibitem[Gol11]{DBLP:books/sp/goldreich2011/Goldreich11g}
Oded Goldreich.
\newblock In a world of p=bpp.
\newblock In Oded Goldreich, editor, {\em Studies in Complexity and
  Cryptography. Miscellanea on the Interplay between Randomness and Computation
  - In Collaboration with Lidor Avigad, Mihir Bellare, Zvika Brakerski, Shafi
  Goldwasser, Shai Halevi, Tali Kaufman, Leonid Levin, Noam Nisan, Dana Ron,
  Madhu Sudan, Luca Trevisan, Salil Vadhan, Avi Wigderson, David Zuckerman},
  volume 6650 of {\em Lecture Notes in Computer Science}, pages 191--232.
  Springer, 2011.

\bibitem[GPS08]{DBLP:journals/dam/GalbraithPS08}
Steven~D. Galbraith, Kenneth~G. Paterson, and Nigel~P. Smart.
\newblock Pairings for cryptographers.
\newblock {\em Discret. Appl. Math.}, 156(16):3113--3121, 2008.

\bibitem[Gro10]{DBLP:conf/asiacrypt/Groth10}
Jens Groth.
\newblock Short pairing-based non-interactive zero-knowledge arguments.
\newblock In Masayuki Abe, editor, {\em Advances in Cryptology - {ASIACRYPT}
  2010 - 16th International Conference on the Theory and Application of
  Cryptology and Information Security, Singapore, December 5-9, 2010.
  Proceedings}, volume 6477 of {\em Lecture Notes in Computer Science}, pages
  321--340. Springer, 2010.

\bibitem[Gro16]{DBLP:conf/eurocrypt/Groth16}
Jens Groth.
\newblock On the size of pairing-based non-interactive arguments.
\newblock In Marc Fischlin and Jean{-}S{\'{e}}bastien Coron, editors, {\em
  Advances in Cryptology - {EUROCRYPT} 2016 - 35th Annual International
  Conference on the Theory and Applications of Cryptographic Techniques,
  Vienna, Austria, May 8-12, 2016, Proceedings, Part {II}}, volume 9666 of {\em
  Lecture Notes in Computer Science}, pages 305--326. Springer, 2016.

\bibitem[GRS99]{DBLP:journals/cacm/GoldschlagRS99}
David~M. Goldschlag, Michael~G. Reed, and Paul~F. Syverson.
\newblock Onion routing.
\newblock {\em Commun. {ACM}}, 42(2):39--41, 1999.

\bibitem[GS89]{DBLP:journals/acr/GoldwasserS89}
Shafi Goldwasser and Michael Sipser.
\newblock Private coins versus public coins in interactive proof systems.
\newblock {\em Advances in Computing Research}, 5:73--90, 1989.

\bibitem[GWC19]{DBLP:journals/iacr/GabizonWC19}
Ariel Gabizon, Zachary~J. Williamson, and Oana Ciobotaru.
\newblock {PLONK:} permutations over lagrange-bases for oecumenical
  noninteractive arguments of knowledge.
\newblock {\em {IACR} Cryptol. ePrint Arch.}, 2019:953, 2019.

\bibitem[HBHW16]{zcash-specs}
Daira Hopwood, Sean Bowe, Taylor Hornby, and Nathan Wilcox.
\newblock Zcash protocol specification.
\newblock
  \url{https://github.com/zcash/zips/blob/master/protocol/protocol.pdf}, 2016.
\newblock Accessed: 2020-06-15.

\bibitem[HG20]{DBLP:journals/iacr/HousniG20}
Youssef~El Housni and Aurore Guillevic.
\newblock Optimized and secure pairing-friendly elliptic curves suitable for
  one layer proof composition.
\newblock {\em {IACR} Cryptol. ePrint Arch.}, 2020:351, 2020.

\bibitem[IW97]{DBLP:conf/stoc/ImpagliazzoW97}
Russell Impagliazzo and Avi Wigderson.
\newblock \emph{P = BPP} if \emph{E} requires exponential circuits:
  Derandomizing the {XOR} lemma.
\newblock In Frank~Thomson Leighton and Peter~W. Shor, editors, {\em
  Proceedings of the Twenty-Ninth Annual {ACM} Symposium on the Theory of
  Computing, El Paso, Texas, USA, May 4-6, 1997}, pages 220--229. {ACM}, 1997.

\bibitem[Jor18]{sharding-blog-post}
Raul Jordan.
\newblock How to scale ethereum: Sharding explained.
\newblock
  \url{https://medium.com/prysmatic-labs/how-to-scale-ethereum-sharding-explained-ba2e283b7fce},
  2018.
\newblock Accessed: 2020-05-27.

\bibitem[Kil92]{DBLP:conf/stoc/Kilian92}
Joe Kilian.
\newblock A note on efficient zero-knowledge proofs and arguments (extended
  abstract).
\newblock In S.~Rao Kosaraju, Mike Fellows, Avi Wigderson, and John~A. Ellis,
  editors, {\em Proceedings of the 24th Annual {ACM} Symposium on Theory of
  Computing, May 4-6, 1992, Victoria, British Columbia, Canada}, pages
  723--732. {ACM}, 1992.

\bibitem[KV20]{cryptoeprint:2020:811}
Markulf Kohlweiss and Mikhail Volkhov.
\newblock Groth16 snarks are randomizable and (weakly) simulation extractable.
\newblock Cryptology ePrint Archive, Report 2020/811, 2020.
\newblock \url{https://eprint.iacr.org/2020/811}.

\bibitem[Lab19]{zk-sync}
Matter Labs.
\newblock zksync: scaling and privacy engine for ethereum.
\newblock \url{https://github.com/matter-labs/zksync}, 2019.

\bibitem[Lip12]{DBLP:conf/tcc/Lipmaa12}
Helger Lipmaa.
\newblock Progression-free sets and sublinear pairing-based non-interactive
  zero-knowledge arguments.
\newblock In Ronald Cramer, editor, {\em Theory of Cryptography - 9th Theory of
  Cryptography Conference, {TCC} 2012, Taormina, Sicily, Italy, March 19-21,
  2012. Proceedings}, volume 7194 of {\em Lecture Notes in Computer Science},
  pages 169--189. Springer, 2012.

\bibitem[LM07]{DBLP:conf/ima/LawM07}
Laurie Law and Brian~J. Matt.
\newblock Finding invalid signatures in pairing-based batches.
\newblock In Steven~D. Galbraith, editor, {\em Cryptography and Coding, 11th
  {IMA} International Conference, Cirencester, UK, December 18-20, 2007,
  Proceedings}, volume 4887 of {\em Lecture Notes in Computer Science}, pages
  34--53. Springer, 2007.

\bibitem[Mat09]{DBLP:conf/pkc/Matt09}
Brian~J. Matt.
\newblock Identification of multiple invalid signatures in pairing-based
  batched signatures.
\newblock In Stanislaw Jarecki and Gene Tsudik, editors, {\em Public Key
  Cryptography - {PKC} 2009, 12th International Conference on Practice and
  Theory in Public Key Cryptography, Irvine, CA, USA, March 18-20, 2009.
  Proceedings}, volume 5443 of {\em Lecture Notes in Computer Science}, pages
  337--356. Springer, 2009.

\bibitem[Mic94]{DBLP:conf/focs/Micali94}
Silvio Micali.
\newblock {CS} proofs (extended abstracts).
\newblock In {\em 35th Annual Symposium on Foundations of Computer Science,
  Santa Fe, New Mexico, USA, 20-22 November 1994}, pages 436--453. {IEEE}
  Computer Society, 1994.

\bibitem[MNT01]{Miyaji2001NewEC}
Atsuko Miyaji, Masaki Nakabayashi, and S.~Takano.
\newblock New explicit conditions of elliptic curve traces for fr-reduction.
\newblock 2001.

\bibitem[Nak09]{bitcoin-whitepaper}
Satoshi Nakamoto.
\newblock Bitcoin: A peer-to-peer electronic cash system.
\newblock \url{http://www.bitcoin.org/bitcoin.pdf}, 2009.

\bibitem[NYM19]{nym-lightpaper}
NYM.
\newblock The nym network, the next generation of privacy infrastructure.
\newblock \url{https://nymtech.net/nym-litepaper.pdf}, 2019.
\newblock Accessed: 2020-05-27.

\bibitem[OABS19]{10.1145/3340422.3343640}
Kai Otsuki, Yusuke Aoki, Ryohei Banno, and Kazuyuki Shudo.
\newblock Effects of a simple relay network on the bitcoin network.
\newblock In {\em Proceedings of the Asian Internet Engineering Conference},
  AINTEC ’19, page 41–46, New York, NY, USA, 2019. Association for
  Computing Machinery.

\bibitem[Ore87]{DBLP:conf/focs/Oren87}
Yair Oren.
\newblock On the cunning power of cheating verifiers: Some observations about
  zero knowledge proofs (extended abstract).
\newblock In {\em 28th Annual Symposium on Foundations of Computer Science, Los
  Angeles, California, USA, 27-29 October 1987}, pages 462--471. {IEEE}
  Computer Society, 1987.

\bibitem[PH09]{Pfitzmann09aterminology}
Andreas Pfitzmann and Marit Hansen.
\newblock A terminology for talking about privacy by data minimization:
  Anonymity, unlinkability, undetectability, unobservability, pseudonymity, and
  identity management, 2009.

\bibitem[PHE{\etalchar{+}}17]{DBLP:conf/uss/PiotrowskaHEMD17}
Ania~M. Piotrowska, Jamie Hayes, Tariq Elahi, Sebastian Meiser, and George
  Danezis.
\newblock The loopix anonymity system.
\newblock In Engin Kirda and Thomas Ristenpart, editors, {\em 26th {USENIX}
  Security Symposium, {USENIX} Security 2017, Vancouver, BC, Canada, August
  16-18, 2017}, pages 1199--1216. {USENIX} Association, 2017.

\bibitem[PMPS00]{DBLP:conf/pkc/PastuszakMS00}
Jaroslaw Pastuszak, Dariusz Michatek, Josef Pieprzyk, and Jennifer Seberry.
\newblock Identification of bad signatures in batches.
\newblock In Hideki Imai and Yuliang Zheng, editors, {\em Public Key
  Cryptography, Third International Workshop on Practice and Theory in Public
  Key Cryptography, {PKC} 2000, Melbourne, Victoria, Australia, January 18-20,
  2000, Proceedings}, volume 1751 of {\em Lecture Notes in Computer Science},
  pages 28--45. Springer, 2000.

\bibitem[PS96]{DBLP:conf/eurocrypt/PointchevalS96}
David Pointcheval and Jacques Stern.
\newblock Security proofs for signature schemes.
\newblock In Ueli~M. Maurer, editor, {\em Advances in Cryptology - {EUROCRYPT}
  '96, International Conference on the Theory and Application of Cryptographic
  Techniques, Saragossa, Spain, May 12-16, 1996, Proceeding}, volume 1070 of
  {\em Lecture Notes in Computer Science}, pages 387--398. Springer, 1996.

\bibitem[Rab83]{MR731318}
Michael~O. Rabin.
\newblock Transaction protection by beacons.
\newblock {\em J. Comput. System Sci.}, 27(2):256--267, 1983.

\bibitem[RZ19]{DBLP:journals/corr/abs-1904-00905}
Antoine Rondelet and Michal Zajac.
\newblock {ZETH:} on integrating zerocash on ethereum.
\newblock {\em CoRR}, abs/1904.00905, 2019.

\bibitem[Sah99]{DBLP:conf/focs/Sahai99}
Amit Sahai.
\newblock Non-malleable non-interactive zero knowledge and adaptive
  chosen-ciphertext security.
\newblock In {\em 40th Annual Symposium on Foundations of Computer Science,
  {FOCS} '99, 17-18 October, 1999, New York, NY, {USA}}, pages 543--553. {IEEE}
  Computer Society, 1999.

\bibitem[Sch74]{MR360598}
Wolfgang~M. Schmidt.
\newblock A lower bound for the number of solutions of equations over finite
  fields.
\newblock {\em J. Number Theory}, 6:448--480, 1974.

\bibitem[Sch80]{MR594695}
J.~T. Schwartz.
\newblock Fast probabilistic algorithms for verification of polynomial
  identities.
\newblock {\em J. Assoc. Comput. Mach.}, 27(4):701--717, 1980.

\bibitem[SCO{\etalchar{+}}01]{DBLP:conf/crypto/SantisCOPS01}
Alfredo~De Santis, Giovanni~Di Crescenzo, Rafail Ostrovsky, Giuseppe Persiano,
  and Amit Sahai.
\newblock Robust non-interactive zero knowledge.
\newblock In Joe Kilian, editor, {\em Advances in Cryptology - {CRYPTO} 2001,
  21st Annual International Cryptology Conference, Santa Barbara, California,
  USA, August 19-23, 2001, Proceedings}, volume 2139 of {\em Lecture Notes in
  Computer Science}, pages 566--598. Springer, 2001.

\bibitem[Sem19]{poma-attack}
Roman Semenov.
\newblock Vulnerability allowing double spend.
\newblock \url{https://github.com/appliedzkp/semaphore/issues/16}, 2019.

\bibitem[Sha92]{DBLP:journals/jacm/Shamir92}
Adi Shamir.
\newblock {IP} = {PSPACE}.
\newblock {\em J. {ACM}}, 39(4):869--877, 1992.

\bibitem[She92]{DBLP:journals/jacm/Shen92}
Alexander Shen.
\newblock {IP} = {PSPACE:} simplified proof.
\newblock {\em J. {ACM}}, 39(4):878--880, 1992.

\bibitem[Sho04]{DBLP:journals/iacr/Shoup04}
Victor Shoup.
\newblock Sequences of games: a tool for taming complexity in security proofs.
\newblock {\em {IACR} Cryptol. ePrint Arch.}, 2004:332, 2004.

\bibitem[SP06]{DBLP:journals/pieee/SampigethayaP06}
Krishna Sampigethaya and Radha Poovendran.
\newblock A survey on mix networks and their secure applications.
\newblock {\em Proceedings of the {IEEE}}, 94(12):2142--2181, 2006.

\bibitem[SS11]{DBLP:journals/em/SilvermanS11}
Joseph~H. Silverman and Katherine~E. Stange.
\newblock Amicable pairs and aliquot cycles for elliptic curves.
\newblock {\em Experimental Mathematics}, 20(3):329--357, 2011.

\bibitem[Tan96]{Tanaka_1996}
Tatsuo Tanaka.
\newblock Possible economic consequences of digital cash.
\newblock {\em First Monday}, 1(2), Aug. 1996.

\bibitem[Teb20]{duncan-attack}
Duncan Tebbs.
\newblock Security of application parameters and nested proofs.
\newblock Private communications, 2020.

\bibitem[Val08]{DBLP:conf/tcc/Valiant08}
Paul Valiant.
\newblock Incrementally verifiable computation or proofs of knowledge imply
  time/space efficiency.
\newblock In Ran Canetti, editor, {\em Theory of Cryptography, Fifth Theory of
  Cryptography Conference, {TCC} 2008, New York, USA, March 19-21, 2008},
  volume 4948 of {\em Lecture Notes in Computer Science}, pages 1--18.
  Springer, 2008.

\bibitem[Vla19]{eip1962}
Alex Vlasov.
\newblock Eip 1962: Ec arithmetic and pairings with runtime definitions.
\newblock \url{https://eips.ethereum.org/EIPS/eip-1962}, 2019.
\newblock Accessed: 2020-05-27.

\bibitem[Woo]{yellow-paper}
Gavin Wood.
\newblock Ethereum: A secure decentralised generalised transaction ledger.
\newblock \url{https://ethereum.github.io/yellowpaper/paper.pdf}.
\newblock Last accessed; 03/08/2020.

\bibitem[WZC{\etalchar{+}}18]{DBLP:conf/uss/WuZCPS18}
Howard Wu, Wenting Zheng, Alessandro Chiesa, Raluca~Ada Popa, and Ion Stoica.
\newblock {DIZK:} {A} distributed zero knowledge proof system.
\newblock In William Enck and Adrienne~Porter Felt, editors, {\em 27th {USENIX}
  Security Symposium, {USENIX} Security 2018, Baltimore, MD, USA, August 15-17,
  2018}, pages 675--692. {USENIX} Association, 2018.

\bibitem[Zip79]{MR575692}
Richard Zippel.
\newblock Probabilistic algorithms for sparse polynomials.
\newblock In {\em Symbolic and algebraic computation ({EUROSAM} '79,
  {I}nternat. {S}ympos., {M}arseille, 1979)}, volume~72 of {\em Lecture Notes
  in Comput. Sci.}, pages 216--226. Springer, Berlin-New York, 1979.

\end{thebibliography}
